\documentclass[manuscript,nonacm]{acmart}

\usepackage{amsfonts,amsmath,amsthm}

\usepackage{booktabs}
\usepackage[utf8]{inputenc}
\usepackage{mathtools}
\usepackage{xspace}
\usepackage{bm}
\usepackage{framed}
\usepackage[capitalise]{cleveref}
\usepackage{enumitem}
\usepackage{algorithm}
\usepackage[noend]{algpseudocode}
\usepackage{comment}

\newtheorem{theorem}{Theorem}
\newtheorem{lemma}{Lemma}

\title{Reaching Consensus for Asynchronous Distributed Key Generation}

\date{\today}
\author{Ittai Abraham}
\affiliation{
\institution{VMware Research}
\city{Herzliya}
\country{Israel}}
\author{Philipp Jovanovic}
\affiliation{\institution{University College London}
\city{London}
\country{United Kingdom}}
\author{Mary Maller}
\affiliation{\institution{Ethereum Foundation}
\city{London}
\country{United Kingdom}}
\author{Sarah Meiklejohn}
\affiliation{\institution{University College London}
\city{London}
\country{United Kingdom}}
\affiliation{\institution{Google}
\city{London}
\country{United Kingdom}}
\author{Gilad Stern}
\affiliation{
\institution{The Hebrew University in Jerusalem}
\city{Jerusalem}
\country{Israel}}
\author{Alin Tomescu}
\affiliation{
\institution{VMware Research}
\city{Palo Alto, CA}
\country{USA}}

\newcommand{\mary}[1]{}
\newcommand{\philipp}[1]{}
\newcommand{\alin}[1]{}
\newcommand{\ittai}[1]{}
\newcommand{\gilad}[1]{}
\newcommand{\sarah}[1]{}
\newcommand{\tocite}[1]{}

\newcommand{\proposal}{\mathsf{prop}}

\newcommand{\validate}{\mathsf{validate}}

\newcommand{\Verify}{\mathsf{Verify}}
\newcommand{\Gather}{\mathsf{Gather}}
\newcommand{\GatherVerify}{\mathsf{GatherVerify}}
\newcommand{\PE}{\mathsf{PE}}
\newcommand{\PEVerify}{\mathsf{PEVerify}}
\newcommand{\AsyncHS}{\mathsf{NWH}}
\newcommand{\keyCorrect}{\mathsf{keyCorrect}}
\newcommand{\lockCorrect}{\mathsf{lockCorrect}}
\newcommand{\commitCorrect}{\mathsf{commitCorrect}}
\newcommand{\viewChange}{\mathsf{viewChange}}
\newcommand{\processErrors}{\mathsf{processFaults}}
\newcommand{\processMessages}{\mathsf{processMessages}}
\newcommand{\agg}{\mathsf{vrf}\_\mathsf{dkg}}

\newcommand{\checkValidity}{\mathsf{checkValidity}}

\newcommand{\generateShare}{\mathsf{\DKG\Share}}
\newcommand{\aggregateShares}{\mathsf{DKGAggregate}}
\newcommand{\DKGVerify}{\mathsf{\DKG\Verify}}
\newcommand{\verifySecret}{\mathsf{\DKG\Share\Verify}}

\newcommand{\sign}{\mathsf{sign}}
\newcommand{\verifySignature}{\mathsf{verifySignature}}

\newcommand{\Eval}{\mathsf{Eval}}
\newcommand{\EvalShareVerify}{\mathsf{EvalShareVerify}}
\newcommand{\evalShare}{\mathsf{eval\_share}}
\newcommand{\ADKG}{\mathsf{ADKG}}
\newcommand{\dkgShares}{\mathsf{dkg\_shares}}
\newcommand{\evalShares}{\mathsf{eval\_shares}}
\newcommand{\evals}{\mathsf{evals}}
\newcommand{\startEval}{\mathsf{start\_eval}}
\newcommand{\evaluation}{\mathsf{evaluation}}
\newcommand{\gatherTuple}[1]{(#1,(\proposal_{#1},\agg_{#1}))}

\newcommand{\randpick}{\xleftarrow{\$}}
\newcommand{\secp}{\lambda}

\newcommand{\Share}{\mathsf{Sh}}
\newcommand{\Aggregate}{\mathsf{Aggregate}}
\newcommand{\share}{\mathsf{share}}

\newcommand{\DKG}{\mathsf{DKG}}

\newcommand{\pk}{\mathsf{pk}}

\newcommand{\sk}{\mathsf{sk}}

\newcommand{\Prove}{\mathsf{Prove}}

\newcommand{\vrf}{\mathsf{vrf}}

\newcommand{\Commit}{\mathsf{Commit}}

\newcommand{\Open}{\mathsf{Open}}

\newcommand{\dkg}{\mathsf{dkg}}

\newtheorem{definition}{Definition}

\algdef{SN}{Upon}{EndUpon}[1]{\textbf{upon} #1, \textbf{do}}
\algdef{SN}{AsLongAs}{EndAsLongAs}[1]{\textbf{as long as} #1, \textbf{run}}

\acmBooktitle{}
\begin{document}

\begin{abstract}

    We give a protocol for Asynchronous Distributed Key Generation (A-DKG) that is optimally resilient (can withstand $f<\frac{n}{3}$ faulty parties), has a constant expected number of rounds, has $\tilde{O}(n^3)$ expected communication complexity, and assumes only the existence of a PKI.
    Prior to our work, the best A-DKG protocols required $\Omega(n)$ expected number of rounds, and $\Omega(n^4)$ expected communication.
    
    Our A-DKG protocol relies on several building blocks that are of independent interest.
    We define and design a \textit{Proposal Election (PE)} protocol that allows parties to retrospectively agree on a valid \textit{proposal} after enough proposals have been sent from different parties.
    With constant probability the elected proposal was proposed by a nonfaulty party.
    In building our PE protocol, we design a \textit{Verifiable Gather} protocol which allows parties to communicate which proposals they have and have not seen in a verifiable manner.
    The final building block to our A-DKG is a \textit{Validated Asynchronous Byzantine Agreement (VABA)} protocol.
    We use our PE protocol to construct a VABA protocol that does not require leaders or an asynchronous DKG setup.  Our VABA protocol can be used more generally when it is not possible to use threshold signatures.
    
\end{abstract}

\maketitle

\renewcommand{\shortauthors}{Abraham et al.}

\section{Introduction}\label{sec:introduction}

In this work we study \textit{Decentralized Key Generation} in the \textit{Asynchronous} setting (A-DKG). Our protocol works in the authenticated model, assumes a Public Key Infrastructure (PKI),  obtains optimal resilience (i.e., tolerates $f<\frac{n}{3}$ malicious parties),
and terminates in $O(1)$ expected rounds using just $\tilde{O}(n^3)$ expected \textit{words}, where a word can contain a constant number of values and cryptographic signatures.
Previously, the best protocol for A-DKG with optimal resilience is by Kokoris-Kogias, Malkhi, and Spiegelman~\cite{KokorisMS20} and it requires $\Omega(n)$ expected number of rounds and $\Omega(n^4)$ expected number of words. 

A DKG protocol allows a set of $n$ parties to collectively generate a public key such that its corresponding secret key is secret-shared between all $n$ parties.  Actions that require the secret key such as decrypting or signing can be performed by any $f+1$ cooperating parties but not by $f$ or fewer.  Unlike in secret sharing protocols, there is no trusted dealer.  
Two key applications of DKGs are threshold encryption and threshold signature schemes.
Threshold encryption can be used to restrict employees' access to databases or to decrypt election results. Threshold signatures can be used to implement random beacons~\cite{syta17scalable}, reduce the complexity of consensus algorithms~\cite{AbrahamMS19}, or more recently to outsource management of secrets on a public blockchain to multiple, semi-trusted authorities~\cite{KokorisAG18}.
One of the challenges in constructing a DKG is that there might be multiple DKG transcripts that would pass verification, and parties must agree on which DKG transcript to eventually use in their application.
This ultimately boils down to a consensus problem in which no preprocessing is possible.
In this work, we are interested in improving the consensus layer of DKG protocols. 
We are careful to avoid the use of any primitive that requires reaching agreement on the output of a DKG (e.g., threshold signatures) in order to instantiate our consensus algorithm.

Kate, Huang, and Goldberg~\cite{KHG12}
observed in an influential paper that many DKGs are unsuitable for use over the Internet due to their reliance on synchrony assumptions and time-outs.
Unstable communication channels are common over the Internet and it is hard to be certain that all players in the system will have seen all messages before moving onto the next round.
Kate, Huang and Goldberg~\cite{KHG12} presented a weakly-synchronous DKG with $O(n^4)$ complexity.
However, their solution relies heavily on leaders who may be adaptively targeted,
and they still require time-outs to distinguish optimistic scenarios from worst-case scenarios.
Recently Kokoris-Kogias, Malkhi and Spiegelman \cite{KokorisMS20} presented a fully asynchronous solution which is leaderless and has $O(n^4)$ expected communication complexity.
The actions of honest parties in their protocol are event-driven and there are no timeouts.

In this work, we are able to improve on the results of Kokoris-Kogias et al.
We design a fully asynchronous consensus algorithm for reaching agreement on the outcome of a DKG that is leaderless and has $\tilde{O}(n^3)$ complexity.
Our solution is secure under the presence of Byzantine adversaries that may corrupt fewer than $\frac{n}{3}$ parties.
Our results are achieved without the use of binary agreements, which is one of the reasons why we are able to improve complexity.
We see this as an important improvement in the design of DKGs that are suitable for use over the Internet as well as a small step towards removing the ``slow'' connotation from the word ``asynchronous''.

\subsection{Our Contributions:}
Our primary contributions are as follows:
\begin{itemize}
    \item Assuming a PKI setup, we present a protocol for solving Asynchronous Distributed Key Generation, that is resilient to $f < \frac{n}{3}$ Byzantine parties, and runs in expected $O(1)$ rounds, where the non-faulty parties send an expected $\tilde{O}(n^3)$ words.
    
    \item We present a new  \textit{Validated Asynchronous Byzantine Agreement} (VABA) protocol that uses a PKI but does \textit{not} use a DKG. Our new VABA protocol can reach agreement on inputs of size $m$ words, in $O(1)$ expected rounds, using just $\tilde{O}(m n^2 +n^3)$ expected words, and is resilient to an adversary controlling at most $f<\frac{n}{3}$ parties.
    Our VABA protocol is the key building block in obtaining our A-DKG.
    
    \item We define and instantiate a new primitive which we call a \textit{Proposal Election} (PE) protocol.
    Our proposal election allows us to avoid relying on leaders.
    Roughly speaking, in Proposal Election, every party inputs some externally valid value and, with constant probability, all parties output the same value that was proposed by a non-faulty party.
    Our Proposal Election runs in $O(1)$ rounds and $\tilde{O}(n^3)$ words and is the key building block in obtaining our VABA protocol.
    
    \item We define and instantiate an extension of the Gather primitive by Canetti and Rabin~\cite{feldman1988optimal,CR93, AAD04} to a \textit{Verifiable Gather} protocol. Our verifiable gather protocol guarantees the existence of some core set, such that all parties output some verifiable super set of this core.  To limit the adversary, only outputs that contains this core pass verification. 
    Our verifiable gather is the key building block in obtaining our proposal election.
\end{itemize}

\subsection{Our techniques}
We obtain our A-DKG using a combination of two advances. 
The first is an \textit{Aggregatable Publicly Verifiable Secret Sharing} (APVSS) scheme by Gurkan \text{et al.}~\cite{GurkanJMMST21} that uses a PKI. 
The second is a \textit{Validated Asynchronous Byzantine Agreement} (VABA) protocol (as defined by Cachin,  Kursawe, Petzold, and Shoup \cite{CachinKPS01}) that uses a PKI but does not use a DKG, which is new to this paper.  
Without a DKG, all previous constant expected time agreement protocols had to rely on a \textit{weak} abstraction (that has a \textit{constant} probability of error) of coin tossing: Feldman and Micali for synchrony \cite{FeldmanM97} and Canetti and Rabin for asynchrony \cite{CR93}.
Our work is also based on this paradigm of using a weak building block. At first sight it may seem that $O(n^4)$ words is the best one can hope for in this paradigm. To obtain an A-DKG with expected $\tilde{O}(n^3)$ word complexity, we identify three barriers, which this work overcomes using novel techniques.

\paragraph{First barrier: aggregate many secret sharings.} Even in  synchronous settings, the weak coin of \cite{FeldmanM97} requires at least $n-f$ parties, such that each such party has at least $f+1$ secrets to be attached to it. If each secret requires a separate \textit{Verifiable Secret Sharing} (VSS) invocation,  we get  $\Omega((f+1)(n-f) |VSS|) = \Omega(n^2 |VSS|)$ word complexity where $|VSS|$ is the word complexity of VSS.
Since VSS, whether asynchronous or not, requires $|VSS|=\Omega(n^2)$ words~\cite{BackesDK13,DolevR82}, we get $\Omega(n^4)$ just to attach enough secrets to enough parties. To overcome this barrier we use an Aggregatable PVSS~\cite{GurkanJMMST21}, which allows to attach $\Omega(n)$ secrets to $\Omega(n)$ parties using just $O(n)$ Reliable Broadcasts \cite{Bracha84,CachinT05a} of $O(n)$-sized APVSS transcripts for a total of $\tilde{O}(n^3)$ word complexity.

\paragraph{Second barrier: Weak Common Coin is too weak.}
Suppose every party can have a random secret sharing attached to it using a total of $\tilde{O}(n^3)$ words. In the classic Binary Asynchronous Byzantine Agreement protocol, these secrets are translated to a weak binary common coin and this coin is used to break ties in case that not all parties have the same input.
The challenge for a VABA protocol aiming for $O(1)$ expected time is the need to randomly elect an externally valid proposal with constant probability. Using a weak common coin to do this election seems challenging. Consider the case where the externally valid inputs are $O(n)$ bits long. We do not know of any way to elect a valid proposal with constant probability using a weak common coin (for example, one could use $\log n$ coins to elect a leader, but due to the constant error probability this will have an error probability that is polynomially close to one).

We suggest a new approach that bypasses the weak coin abstraction. Instead, we proceed to extend the Gather primitive of Canetti and Rabin \cite{feldman1988optimal,CR93,AAD04} to a \textit{Verifiable} Gather protocol. Recall that a Gather protocol does not solve consensus but instead guarantees the existence of some core set, such that all parties output some super set of this core. Roughly speaking, the goal of our new Verifiable Gather primitive is to introduce a verification protocol to essentially force the adversary to also only output super sets of this core (in the sense that other outputs will not pass the verification).

We show how to combine Verifiable Gather with random secret sharing \cite{KokorisMS20} and an efficient Reliable Broadcast \cite{Bracha84, Bracha87,CachinT05a} to obtain a new primitive we call Proposal Election. Roughly speaking, in Proposal Election, every party inputs some externally valid value, and with constant probability, all parties output the same value that was proposed by a non-faulty party.
Our Proposal Election runs in $O(1)$ rounds and $\tilde{O}(n^3)$ words.  

Conceptually, our Proposal Election abstraction can be viewed as the validated (multi-valued) generalization of the weak common coin approach. Technically, our Proposal Election (PE) exposes a new validation abstraction that efficiently enables electing a common externally valid value with constant probability. 
Crucially, parties can also verify that other parties provide the uniquely elected value if the election process succeeded. This significantly limits the adversary's behaviour and forces it to essentially act honestly or remain silent.

\paragraph{Third barrier: efficient VABA, using PE} Our final challenge for asynchronous DKG is obtaining a VABA protocol for messages of size $m$ (where $m=\Theta(n)$ words, is the size of a PVSS) using PE at a cost of just $\tilde{O}(mn^2 + n^3)=\tilde{O}(n^3)$ words per view and just $O(1)$ expected views (due to the constant success probability of PE), where each view consists of just a constant number of rounds.
There are two natural approaches. The first is to use known optimally resilient \textit{validated multi-valued} techniques from known  VABA protocols. Unfortunately, the known VABA protocols of Cachin, Kursawe, and Shoup~\cite{CachinKS05} and Abraham, Malkhi, and Spiegelman~\cite{AbrahamMS19} require a DKG where all parties agree on the output (except for negligible error) and do not seem to work with the constant error probability of PE. The work of Cachin, Kursawe, Lysyanskaya and Strobl~\cite{CKLS02} uses an existing DKG to refresh to a new DKG using $\Omega(n^4)$ words. The work of Zhou, Schneider and Van Renesse~\cite{APSS} suggest a refresh protocol with exponentially high communication complexity.

The second natural approach is to use  \textit{binary} agreement techniques. Indeed, the application of Bracha's consensus technique \cite{Bracha84} (with our PE protocol) requires $\Omega(n)$ invocations of Reliable Broadcast per bit, for a total of $\Omega(mn^3)=\Omega(n^4)$ words when $m=\Omega(n)$ (and this solution only obtains weak validity).

We overcome this third barrier with a new consensus protocol called \textit{No Waitin' HotStuff} (NWH). As its name implies, NWH is a new member of the HotStuff family of consensus protocols \cite{HS19, AbrahamMS19,DUMBO20,AbrahamS20} which obtains $\tilde{O}(n^3 +m n^2)$ expected words and $O(1)$ expected rounds in the asynchronous setting, using PE, and without relying on a DKG.

Intuitively, in each view of NWH, a new invocation of PE is used as a "virtual leader". 
For safety, NWH uses the by-now-standard \textit{Key-Lock-Commit} paradigm of HotStuff \cite{HS19,AbrahamMS19}. 
The main novelty of NWH is in its liveness guarantees and its ability to change view in asynchrony in a constant number of asynchronous rounds even if the "virtual leader" acts maliciously.
NWH obtains liveness in full asynchrony using our PE's properties and a new mechanism that forces parties (even malicious parties) to essentially send only validated responses. 
In case of a non-faulty "virtual leader", the PE properties guarantee that all non-faulty parties see the \textit{same} output from the leader and that this input was an input of a non-faulty party. In this case, the NWH protocol forces the faulty parties to essentially only act as omission-faulty (hence a decision is guaranteed to be reached in such a view).
In case of a faulty "virtual leader", the PE properties guarantee that all non-faulty parties eventually see \textit{some} output from the leader (might not be the same), and the NWH protocol guarantees that only a safe decision will be made or, if none can be reached, eventually a view change will occur in a constant number of rounds. The combination of NWH with the constant probability of success for PE guarantee termination in an expected constant number of asynchronous rounds. 
NWH manages to obtain these safety and liveness properties to obtain a VABA protocol for messages of size $m$ words with $\tilde{O}(mn^2 + n^3)$ expected message complexity and $O(1)$ expected rounds. 

\paragraph{A Note on Adaptive Adversaries} All our results hold for a static adversary. However, we note that given an aggregatable PVSS scheme that is secure against adaptive adversaries, our VABA protocol and therefore our A-DKG protocol would also be secure against adaptive adversaries. This is the same type of reduction as in \cite{CachinKS05,AbrahamMS19} where the protocol is adaptivly secure if its underlying cryptographic primitives are adaptivly secure. The PVSS scheme of \cite{GurkanJMMST21} is only proved security in the static model. Obtaining an adaptively-secure aggregatable PVSS remains an open question. 

\subsection{Related Work}
Our work assumes a PKI and obtains a Validated ABA protocol. However, many of our techniques can be seen as (non-trivial) extensions of the work done in the information theoretic model (where there are private channels, but no PKI nor any computational bounds on the adversary).
In the information theoretic model, the natural validity property is weaker and it is natural to focus on the binary case.
Any solution for consensus in the asynchronous model must have infinite executions~\cite{FLP85}. Ben-Or~\cite{BenO83} showed how randomization can be used to obtain a finite expected running time and Bracha~\cite{Bracha84} showed how to do this with optimal resilience. Reducing the expected number of rounds to a constant was obtained by 
Canetti and Rabin~\cite{CR93}. They provide the first ABBA with optimal resilience and constant expected time. It requires at least $\Omega(n^8)$ words in expectation (possibly more, but we did not verify). This was improved by Patra, Choudhary, and Rangan~\cite{AAR09} to expected $\tilde{O}(n^4)$ words for ABBA.
The protocols of Canetti and Rabin~\cite{CR93}, their extensions and those that rely on cryptographic assumptions all have a non-zero probability of non-termination.
In the information theoretic setting it is possible to efficiently solve \textit{Asynchronous Binary Byzantine Agreement (ABBA)} with optimal resilience and zero probability of non-termination \cite{ADH08}, and this can be done with just $\tilde{O}(n^6)$ expected words and $O(n)$ rounds \cite{BCP18}.

The verifiable weak proposal election primitive is an extension of the idea of a weak common coin, which was introduced in the synchronous setting by Feldman and Micali\cite{FeldmanM97}.
A weak common coin is a primitive simulating a common shared randomness source.
The coin is weak in the sense that with some probability the parties might not agree on the value.
Feldman later extended this result to the asynchronous setting \cite{feldman1988optimal}.
Katz and Koo improve on the synchronous result \cite{KatzK06}.

A DKG can be viewed as a specific form of a Multi-Party Computation (MPC) protocol. In that sense, the work of Ben-Or, Canetti and Goldreich~\cite{BCG93} obtains perfect security for $n>4f$ and the work of Ben-Or, Kelmer and Rabin~\cite{BKR94} obtains statistical security and optimal resilience of $n>3f$. Both protocols use ABBA as a building block and have very high word complexity. 
Modern MPC protocols in the asynchronous model use a DKG~\cite{BH07, HiNiPr08,CP15}, so they could benefit from the results of our work.
Another related work that may benefit from protocol is the work of Gągol, Leśniak, Straszak and Świętek \cite{gagol2019aleph}.

\section{Definitions and Assumptions}\label{sec:definitions}

\subsection{Network and Threat Model} \label{sec:definitions:threatmodel}

This work deals with protocols for $n$ parties with point-to-point communication channels.
The network is assumed to be asynchronous, which means that there is no bound on message delay, but all messages must arrive in finite time.
The protocols below are designed to be secure against a Byzantine adversary controlling up to $f<\frac{n}{3}$ parties.
This work uses several cryptographic assumptions as "perfect" black-boxes, meaning we assume that an adversary cannot break them.
As described in \cite{AbrahamMS19,CachinKPS01,CachinKS05}, with high probability all protocols require polynomially many uses of the cryptographic primitives, so the protocols remain secure in the face of a computationally bounded adversary with all but a negligible probability.
As described in the introduction, the protocols themselves are secure against adaptive adversaries given an instantiation of the cryptographic primitives which is secure against such an adversary.
However, currently there are no known adaptively secure instantiations for all of the primitives we require.
Similar to the protocols of  \cite{AbrahamMS19,CachinKS05}, the protocols presented can be seen as reductions from one task to another that preserve security against adaptive adversaries.

\subsection{Reliable Broadcast} \label{sec:definitions:broadcast}

A \textit{Reliable Broadcast} is an asynchronous protocol with a designated \textit{dealer}. The dealer has some \textit{input value} $M$ from some known domain $\mathcal{M}$ and each party may \textit{output} a value in $\mathcal{M}$. 
A Reliable Broadcast protocol has the following properties assuming all nonfaulty parties participate in the protocol:
    \begin{itemize}
        \item \textbf{Validity.} If the dealer is nonfaulty, then every nonfaulty party that completes the protocol outputs the dealer's input value, $M$. 
        \item \textbf{Agreement.} If two nonfaulty parties output some value, then it's the same value.
        \item \textbf{Termination.} If the dealer is nonfaulty, then all nonfaulty parties complete the protocol and output a value.
        Furthermore, if some nonfaulty party completes the protocol, every nonfaulty party completes the protocol.
    \end{itemize}
A \textit{Validated Reliable Broadcast} protocol is a Reliable Broadcast protocol variant where each party has access to a common \textit{validate} function, $\validate:\mathcal{M} \to \{0,1\}$. We say that $M \in \mathcal{M}$  is \textit{externally valid} if $\validate(M)=1$.
In a Validated Reliable Broadcast protocol, the dealer has an externally valid input. 
A Validated Reliable Broadcast protocol has the following additional property: 
\begin{itemize}
    \item \textbf{External Validity.} If a nonfaulty party outputs a value, then this value is externally valid.
\end{itemize}
See Appendix~\ref{sec:broadcast} for a Reliable Broadcast protocol and a Validated Reliable Broadcast protocol with word complexity of $\tilde{O}(n^2 + mn)$, where $m$ is the number of words in any value in $\mathcal{M}$.

\subsection{Verifiable Gather} \label{sec:definitions:gather}

\textit{Gather} is a natural \textit{multi-dealer} extension of Reliable Broadcast where every party is also a dealer. The output of a gather protocol is a \textit{gather-set}. A gather-set consists of \textit{at least} $n-f$ pairs $(j,x)$, such that $j\in [n]$, $x\in \mathcal{M}$, and  each index $j$ appears at most once.
For any given gather-set $X$, we define its index-set $Indices(X)=\{j|\exists (j,x)\in X\}$ to be the set of indices that appear in $X$.

Intuitively speaking, the goal of Gather is to have some common \textit{core} gather-set such that all parties output a super-set of this core. 
Note that a Gather protocol does not solve consensus and different parties may output different super-sets of the core. 
For \textit{Verifiable Gather}, the goal is to limit the power of the adversary to generate inconsistent outputs. 
Intuitively, for any gather-set produced by the adversary, if it passes some \textit{verification protocol}, it must also be a super-set of the common core.

Formally, a verifiable gather protocol consists of a pair of protocols $(\Gather,\Verify)$ and takes as input an external validity function $\validate$ which all parties have access to.
For $\Gather$, each party $i \in [n]$ has an externally valid \textit{input} $x_i$. 
Each party may decide to \textit{output} a gather-set $X_i$.
After outputting the gather-set, parties must continue to update their local state according to the $\Gather$ protocol in order for the verification protocol to continue working.

The properties of $\Gather$ (assuming all nonfaulty start):

\begin{itemize}
    \item \textbf{Binding Core.} Once the first nonfaulty party outputs a value from the $\Gather$ protocol there exists a core gather-set $X^*$ such that if a nonfaulty party $i$ outputs the gather set $X_i$, then $X^*\subseteq X_i$.
    \item \textbf{Internal Validity.} If $(j,x) \in X^*$ and $j$ is nonfaulty at the time the first nonfaulty party completed the $\Gather$ protocol,
    then $x$ is the input of party $j$ in $\Gather$.
    \item \textbf{Termination of Output.} All nonfaulty parties  eventually output a gather-set.
\end{itemize}

The $\Verify$ protocol receives an index-set $I$ and outputs a gather-set $X$ such that $Indices(X)=I$.
It performs two actions at once: it verifies that the index set includes the indices of the binding core, and recovers the gather-set only from the indices and the internal state of the verifying party.
This allows parties to send relatively small index-sets instead of large gather-sets over the network.
The verification protocol limits the adversary to a very narrow set of behaviours, so that any \textit{verifiable} gather-set must contain the Binding core gather-set $X^*$.
A party $i$ can check any index-set $I$, which we denote by executing $\Verify_i(I)$.
If the execution of $\Verify_i(I)$ terminates and outputs a value, we say that $i$ has verified the index-set $I$.

The termination properties of $\Verify$ (given that all nonfaulty start $\Gather$):
\begin{itemize}
        \item \textbf{Completeness.} For any two nonfaulty parties $i,j$, if $j$ outputs $X_j$ from $\Gather$, then $\Verify_i(Indices(X_j))$ eventually terminates with the output $X_j$.
        
        \item \textbf{Agreement on Verification.} For any two nonfaulty $i,j$, and any index-set $I$, if $\Verify_i(Y)$ terminates with the output $X$ then $\Verify_{j}(I)$ eventually terminates with the output $X$. 
\end{itemize}

The correctness properties of the $\Verify$ protocol:
\begin{itemize}

        \item \textbf{Agreement.} All nonfaulty parties agree on values with common indexes. For any two nonfaulty $i,j$, and any index-sets $I,J$, if $\Verify_i(I)$ terminates with the output $X$ and $\Verify_{j}(J)$ terminates with the output $Y$, and $(k,x) \in X,(k,y) \in Y$, then $x=y$. 
        
        \item \textbf{Includes Core.} If $\Verify_i(I)$ terminates with the output $X$, then the gather-set $X$ contains the binding core gather-set $X^*$ (as defined in the Binding Core property of $\Gather$). 

        \item \textbf{External Validity.} If $\Verify_i(I)$ terminates with the output $X$ for some nonfaulty $i$, then for each $(j,x) \in X$, the value $x$ is externally valid.
\end{itemize}

Observe that the Includes Core and Completeness properties say that not only do all nonfaulty output a gather-set that includes the core but that any gather-set that passes verification contains the core $X^*$.

\subsection{Proposal Election} \label{sec:definitions:VWPE}

A perfect proposal election would allow each party to input a proposal and then have all parties output one common randomly elected proposal. \textit{Proposal Election} (PE) is an asynchronous protocol that tries to capture this spirit but obtains weaker properties. Intuitively, there is only a constant probability that the output  of PE is one common randomly elected proposal coming from a nonfaulty proposer. As in the Verifiable Gather (VG) protocol, we also add a verification protocol. Crucially, in the good event mentioned above, the only value that passes verification is this common elected proposal.
In the remaining cases, the adversary can control the output and even cause different parties to have different outputs. However, even in these cases we force the adversary to allow all parties to eventually output some verifying value.  
This PE is weak enough to be efficiently implementable and we will later show that it is strong enough to enable an efficient constant expected round VABA protocol.

As in VG, we assume a domain $\mathcal{M}$ and we are externally given a function $\validate$ that given any message $x \in \mathcal{M}$ can check the external validity of $x$.
A Proposal Election protocol consists of a pair of protocols $(\PE,\Verify)$.
Each nonfaulty party $i$ starts with an externally valid input $x_i$ to $\PE$.
The output of the $\PE$ protocol is a pair $(x,\pi)$ where $x\in\mathcal{M}$ and $\pi$ is a proof used in the $\Verify$ protocol. 
We model these protocols as having some ideal write-once state $x^*$. We assume $\bot$ is not externally valid and let $x^*\in\mathcal{M}\cup\{\perp\}$. Intuitively, if $x^* \neq \bot$ then the output of all parties will be $x^*$, but when $x^* =\bot$ then the adversary can cause different parties to output different verifying values.

\begin{itemize}
        \item $\alpha$-\textbf{Binding}. 
        For any adversary strategy, with probability $\alpha$, $x^*$ is set to an input of a party that behaved in a nonfaulty manner when it started the $\PE$ protocol.
\end{itemize}
In addition, the $\PE$ protocol has a natural termination property (assuming all nonfaulty start):

\begin{itemize}
        \item \textbf{Termination of Output.} All nonfaulty parties eventually output a pair $(x,\pi)$. 
\end{itemize}

A party $i$ can check any pair of proposal and proof, $(x,\pi)$, which we denote by executing $\Verify_i(x,\pi)$. 
If the execution of $\Verify_i(x,\pi)$ terminates, we say that $i$ has verified $x$. 
If the binding value $x^*$ is not $\perp$, then the only value for which the verify protocol can terminate is $x^*$.
This limits the adversary to essentially either reporting $x^*$, or remaining silent.
The termination properties of $\Verify$ (given that all nonfaulty start $\PE$):

\begin{itemize}
        \item \textbf{Completeness.} For any two nonfaulty $i,j$, the output $(x,\pi)$ of party $j$ from  $\PE$ will eventually be verified by party $i$, i.e. $\Verify_i(x,\pi)$ eventually terminates.
        
        \item \textbf{Agreement on Verification.} For any two nonfaulty $i,j$, and any value $x$ and proof $\pi$, if $\Verify_i(x,\pi)$  terminates then $\Verify_{j}(x,\pi)$ eventually terminates.
\end{itemize}    
Finally, the correctness properties of  $\Verify$:
\begin{itemize}
        \item \textbf{Binding Verification.} If $x^*\neq \perp$ then for every nonfaulty party $j$, and every $(x,\pi)$, if $\Verify_j(x,\pi)$ terminates then $x=x^*$.
        
        \item \textbf{External Validity.} If $\Verify_i(x,\pi)$ terminates then the value $x$ is externally valid.
\end{itemize}
We note that in the computational setting all these properties hold with all but negligible probability.

\subsection{Validated Asynchronous Byzantine Agreement}

In a Validated Asynchronous Byzantine Agreement protocol, there is some external validity function that every party has access to.
In addition, there exists some success parameter $\alpha\in(0,1)$ for the protocol.
Each nonfaulty party $i$ starts with some externally valid input $x_i$ and on termination must output a value.
A Validated Asynchronous Byzantine Agreement protocol has the following properties (assuming all nonfaulty start):
\begin{itemize}
        \item \textbf{Agreement.} All nonfaulty parties that complete the protocol output the same value.
        \item \textbf{Validity.} If a nonfaulty party outputs a value then it is externally valid.
        
        \item \textbf{$\alpha$-Quality.} With probability $\alpha$, the output value is chosen as one of the inputs $x_i$ (party $i$ was nonfaulty when it started the protocol). 
        
        \item \textbf{Termination.} All nonfaulty parties almost-surely terminate, i.e. with probability 1.
    \end{itemize}

\subsection{Cryptographic Abstractions} \label{sec:definitions:crypto}
This work introduces a novel distributed consensus algorithm which uses several cryptographic tools as black-boxes. 
In Section~\ref{sec:efficiency} we discuss how these tools can be instantiated with respect to tools that currently exist in the literature and evaluate the efficiency of our protocol with respect to these tools.
The instantiations of the cryptographic abstractions in this paper are all assumed from prior work, with the exception of an A-DKG protocol, which we define in this section and construct in \cref{sec:dkg}.

\subsubsection{Distributed Key Generation}
A distributed key generation algorithm is a method to generate public keys for threshold systems without a trusted third party.
It is assumed that the aggregation and verification algorithms keep state consisting of each party's public key.
A DKG consists of the following algorithms.
\begin{itemize}
    \item $\generateShare(\sk_i) \mapsto \dkg\share:$ A probabilistic algorithm run by Party $i$ that takes as input a secret key and outputs a \textit{DKG share}.
    The share also contains a description of the party who sent it.
    \item $\generateShare\Verify(\pk_i, \dkg\share) \mapsto \{0,1\}:$ A deterministic algorithm run by Party $j$ that returns $1$ if it is convinced that the DKG share of Party $i$ is valid.
    \item $\DKG\Aggregate(\mathcal{D}) \mapsto \dkg:$ An algorithm run by Party $i$ that takes as input a set $\mathcal{D}$ containing at least $2f + 1$ DKG shares from different parties and outputs a \textit{DKG transcript}.
    \item $\DKG\Verify(\dkg) \mapsto \{0,1\}:$  A deterministic algorithm that returns $1$ if and only if the DKG transcript  contains  DKG shares that pass verification from at least $2f+1$ different parties.
\end{itemize}
The non-inclusion of a reconstruction algorithm here is deliberate; we assume that the purpose of the DKG is to generate a public key for a threshold application and as such it is not clear that a reconstruction algorithm is useful.

A distributed key generation algorithm should be security preserving and correct.  
As the purpose of a distributed key generation algorithm is to generate a public key, secrecy guarantees are only meaningful in the context of the threshold scheme it is being used to instantiate.
Security preservation captures this notion: it means that provided no more than $f$ parties are corrupted, a threshold scheme under the DKG retains all properties of the standard scheme under the key generation algorithm. 
For the sake of this paper we only formally define security preservation for our threshold verifiable random function and instead refer to \cite{GurkanJMMST21} for a full definition of security preservation.

\begin{definition}
    An Asynchronous Distributed Key Generation protocol has the following properties:
    \begin{itemize}
        \item \textbf{Security Preservation.} A threshold scheme under the DKG retains all the properties of the standard scheme under the key generation algorithm, provided no more than $f$ parties are corrupted.
        \item \textbf{Correctness.} We have that:
        \[
        \DKG\Share\Verify(\pk_{i}, \generateShare(\sk_{i}) ) = 1
        \]
        Assume that every $\dkg\share_i \in \mathcal{D}$ is such that
        $\generateShare\Verify(\pk_i, \dkg\share_i) = 1$.
        Then
        \[
        \DKG\Verify( \DKG\Aggregate(\mathcal{D}) ) = 1.
        \]
    \end{itemize}
\end{definition}

An asynchronous DKG, which is the topic of this paper, is an interactive protocol allowing all parties to output the same aggregated DKG transcript.
Since the network is asynchronous, it is also important to make sure that the parties eventually complete the protocol.
Therefore, an A-DKG protocol has the following two properties if all nonfaulty parties participate in it:
\begin{itemize}
    \item \textbf{Agreement.} All parties that terminate output the same DKG , $\dkg$, such that $\DKG\Verify(\dkg)=1$.
    
    \item \textbf{Termination.} All nonfaulty parties almost-surely terminate, i.e. with probability 1.
\end{itemize}

\subsubsection{Threshold Verifiable Random Function} 
A threshold \textit{verifiable random function (VRF)} is an algorithm such that $(f + 1)$ parties can compute the output of the random function $\phi$ on some input,
but $f$ cannot. 
A threshold VRF must be unbiasable ($f$ parties cannot guess even a single bit of the outcome),
and robust ($f + 1$ honest parties always agree on the output).
We will instantiate the threshold VRF using the aggregatable DKG and VUF of Gurkan et al.~\cite{GurkanJMMST21}.

In addition to $(\DKG\Share, \DKG\Share\Verify, \DKG\Aggregate, \DKG\Verify)$ defined above, a threshold VRF consists of the following algorithms:
\begin{itemize}
    \item $\phi(\vrf\_\dkg, m) \mapsto \{0,1\}^{\secp}:$ A deterministic function that takes in a DKG transcript (which implicitly defines a secret key) and a message, and outputs a binary string.  We have that $\phi$ cannot be computed by less than $f+1$ parties.
    \item $\Eval\Share(\vrf\_\dkg, \sk_i,  m) \mapsto (\phi_i(m), \pi_i):$ A probabilistic algorithm run by Party $i$ that takes as input a DKG transcript, a secret key, and a message and returns an \textit{evaluation share} and a \textit{proof share}. Here $\phi_i$ is used to denote that this is a share of $\phi(m)$ as opposed to the full evaluation (likewise $\pi_i$).
    The share also contains a description of the party who sent it.
    \item $\Eval\Share\Verify(\vrf\_\dkg, \pk_i, m, \phi_i(m), \pi_i) \mapsto \{0,1\}:$  A deterministic algorithm run by Party $j$ that takes as input a VRF-DKG transcript, a public key, a message, an evaluation share, and a proof share from Party $i$ and returns $0/1$ to indicate rejection/acceptance.
    \item $\Eval(\vrf\_\dkg, m, \mathcal{F}) \mapsto (\phi(\vrf\_\dkg,m), \pi):$ An algorithm that takes as input a DKG transcript, a message, and a set $\mathcal{F}$ that contains evaluation and proof shares from $f+1$ different parties.  It outputs a function evaluation and an aggregated proof.
    \item $\Eval\Verify(\vrf\_\dkg, m, \phi(\vrf\_\dkg, m), \pi)\mapsto \{0,1\}$:
    A deterministic algorithm that takes as input a DKG transcript, a message, a function evaluation and a proof.  It outputs $0/1$ to indicate rejection/acceptance.
\end{itemize}

\begin{definition}
    A Threshold Verifiable Random Function has the following properties:
    \begin{itemize}
        \item \textbf{Unbiasability.}
        The function $\phi(\vrf\_\dkg,m)$ is distributed uniformly at random over all verifying DKGs and the message space $\mathcal{M}$. 
        Let $\vrf\_\dkg$ be an aggregated DKG transcript such that $\DKG\Verify(\vrf\_\dkg) = 1$.
        Then as long as no nonfaulty party computes $\Eval\Share(\vrf\_\dkg,m)$, then the adversary cannot guess a single bit of $\phi(\vrf\_\dkg,m)$.
        
        \item \textbf{Uniqueness.}
        For each $\vrf\_\dkg, m$, there is a single value $v = \phi(\vrf\_\dkg, m)$ such that
        there exists $\pi$ with \[\Eval\Verify(\vrf\_\dkg, m, v, \pi) = 1.\]
        
        \item \textbf{Correctness.} We have that:
        \[
        \Eval\Share\Verify(\vrf\_\dkg, \pk_i, m, \Eval\Share(\vrf\_\dkg, \sk_i, m)  ) = 1
        \]
         Assume that every $(\phi_i(\vrf\_\dkg, m), \pi_i) \in \mathcal{F}$ is such that 
         $\Eval\Share\Verify(\vrf\_\dkg, \pk_i, m, \phi_i(\vrf\_\dkg, m), \pi_i  ) = 1$.
         Then
         \[
         \Eval\Verify(\vrf\_\dkg, m, \Eval(\vrf\_\dkg, m, \mathcal{F})) = 1.
         \]
    \end{itemize}
\end{definition}
Unbiasability also assumes that no honest party has sent a reconstruction share for $\vrf\_\dkg$.
We have chosen not to explicitly state this in the definition because we have omitted a description of a reconstruction algorithm for the DKG.
When the purpose of the DKG is to generate a public key for a threshold VRF, no reconstruction takes place.

\subsubsection{Vector commitment}
A vector commitment is used to bind a party to a vector, such that they can later provably reveal any position in the vector.
A vector commitment consists of the following algorithms.
\begin{itemize}
    \item $\Commit(v) \mapsto c:$ Takes as input a vector $v$ and outputs a commitment $c$.
    \item $\Open\Prove(c, v, i) \mapsto \pi:$ Takes as input a commitment $c$ to a vector $v$ and an evaluation point $i$. 
    Outputs a proof that the $i$th entry of $v$ is $v_i$.
    \item $\Open\Verify(c, v_i, i, \pi) \mapsto 0/1:$ A deterministic algorithm that takes as input a commitment $c$, an opening $v_i$, an evaluation point $i$ and a proof $\pi$. It outputs $1$ if it is convinced that the $i$th entry of the vector committed in $c$ is $v_i$ and $0$ otherwise. 
\end{itemize}

In this work we only require the vector commitment to satisfy binding i.e. that an adversary cannot open a commitment to more than one value at any evaluation point.
It does not necessarily need to be hiding.
\begin{itemize}
    \item \textbf{Correctness.} $\forall$ vectors $v$, $\forall$ positions $i$, we have
    \[
    \Open\Verify( \Commit(v), v_i, i, \Open\Prove(c, v, i)) = 1.
    \]
    \item \textbf{Binding.} No adversary can compute a commitment $c$, an evaluation point $i$, two values $v_i$ and $w_i$ with $v_i \neq w_i$, and two proofs $\pi_v$ and $\pi_w$ such that
    \[
    \Open\Verify(c, v_i, i, \pi_v) = \Open\Verify(c, w_i, i, \pi_w) = 1.
    \]
\end{itemize}
\section{Verifiable Gather}\label{sec:gather}

As part of our proposal election protocol we require a ``reliable gather''.
Throughout the protocol, parties reliably broadcast values, which are later used to choose a winning proposal from among them.
Ideally, we would like the parties to agree on an exact set of parties and broadcasted values in order to make sure that they all elect a value from the same set.
However, exactly agreeing on the set is non-trivial and potentially expensive.
Therefore we slightly relax our requirements: there exists some core $C$ of size $n-f$ or greater such that the output of every nonfaulty party contains $C$.
Furthermore, we would like parties to be able to prove that they ``acted correctly'' and included $C$ in their output.

Throughout the protocol, parties broadcast messages using the Reliable Broadcast protocol $RB$ and validated broadcast messages using the Validated Reliable Broadcast protocol $VRB$.
In a slightly inaccurate high-level view, the protocol takes place in three rounds. 
In the beginning, all parties broadcast their inputs and wait to receive $n-f$ broadcasts from other parties.
After receiving those broadcasts, they broadcast sets of tuples containing values and the parties who sent them in the previous round.
They then wait to receive $n-f$ such sets, checking if the sets report the correct values.
After receiving $n-f$ of those sets, every party broadcasts the union of all of the reported sets.
Finally, after receiving $n-f$ such unions and checking that the reported sets are correct, every party outputs the union of those sets.
However, when dealing with large inputs, broadcasting sets of $O(n)$ values can be an unnecessarily expensive operation.
In order to avoid this overhead, parties only actually broadcast their values in the first round.
In any subsequent round, parties only refer to the broadcasted value by the party who sent the relevant broadcast, requiring only one word per value.

More accurately the protocol can be broken into three rounds:

\noindent \textbf{Round 1:} In the first round, party $i$ validated broadcasts its input value $x_i$ and waits to receive $n-f$ valid values from all parties.
Party $i$ stores the parties from whom it received broadcasts in a set $S_i$, and tuples of the form $(j,x_j)$ indicating that it received the value $x_j$ from $j$ in a set $R_i$.
\\

\noindent \textbf{Round 2:} After receiving $n-f$ values, each $i$ broadcasts $S_i$, which we think of as sets of the values $x_j$ referenced only by the party who sent each value.
Party $i$ then waits to receive $n-f$ $S$ sets from other parties, and accepts such a message after seeing that it received a value from each party in $S$.
After accepting a message with the set $S$ from $j$, $i$ adds $j$ to $T_i$.
We think of $T_i$ as containing all of the $S$ sets received from different parties, while it actually only references each set by the party who sent it.
\\

\noindent \textbf{Round 3:} Finally, once $T_i$ is of size $n-f$, $i$ broadcasts $T_i$ as well and waits to receive $n-f$ such sets.
Similarly to before, $i$ only accepts a message with a set $T$ if it accepted all of the $S$ messages it refers to.
After accepting a set $T_j$, $i$ explicitly computes the union of all of the $S$ sets $T_j$ is referring to in the following manner: $V_j=\bigcup_{k\in T_j} S_k$, and stores $(j,V_j)$ in $U_i$.
Once $i$ accepts $n-f$ different messages containing $T$ sets and updates $U_i$, it outputs $R_i$ which contains tuples of values and the parties who sent them.
It is important to note that when outputting $R_i$ it contains all of the element in all of the sets referred to by any accepted $T$ set, because parties wait to receive all relevant information before accepting a $T$ or an $S$ set.
Every party continues updating its internal state even after outputting a value.
\\

In the verification protocol for an index-set $I$,  party $i$ checks whether $X$ includes all of the values referred to by at least $n-f$ of the $T$ sets that it received and accepted.
In the following discussion we show that there exists some index $i^*$ that is included in at least $f+1$ of the $T$ sets broadcasted by parties.
Since every party waits to receive $T$ sets from at least $n-f$ parties before terminating, it will see at least one with that index, and thus include $S_{i^*}$ in its output.
This is true for any nonfaulty party, so $S_{i^*}$ can serve as a common-core in the output of all nonfaulty parties.
Similarly, when verifying an index-set $I$, $i$ makes sure that it contains the values referenced by the $T$ sets received from at least $n-f$ parties, and thus also includes $S_{i^*}$ in it.
Afterwards, the values corresponding to each index can easily be returned because they have been previously received by broadcast.

\begin{algorithm}\caption{$\Gather_i(x_i)$}
\begin{algorithmic}[1]
\State $R_i\gets\emptyset,S_i\gets\emptyset,T_i\gets\emptyset,U_i\gets\emptyset$
\State \textbf{validated broadcast} $\langle 1,x_i\rangle$ with external validity function returning $1$ on $\langle t,m\rangle$ iff $\validate(m)=1$
\Upon{receiving $\langle1,x_j\rangle$ from $j$}
    \State $R_i\gets R_i\cup\{(j,x_j)\},S_i\gets S_i\cup\{j\}$
    \If{$\left|S_i\right|=n-f$}
        \State \textbf{broadcast} $\langle 2,S_i \rangle$
    \EndIf
\EndUpon
\Upon{receiving $\langle 2,S_j \rangle$ from $j$ such that $\left|S_j\right|\geq n-f$}
    \Upon{$S_j\subseteq S_i$} 
        \State $T_i\gets T_i\cup \{j\}$
        \If{$\left| T_i\right|=n-f$}
            \State \textbf{broadcast} $\langle 3,T_i \rangle$ \Comment{$T$ sets reference $S$ sets}
        \EndIf
    \EndUpon
\EndUpon
\Upon{receiving $\langle 3,T_j \rangle$ from $j$ such that $\left|T_j\right|\geq n-f$}
    \Upon{$T_j\subseteq T_i$} \Comment{relevant $S$ sets and values are received} 
        \State $U_i\gets U_i\cup \{(j,\bigcup_{k\in T_j} S_k)\}$ \Comment{save all parties in the $S$ sets referenced by $T_j$}
        \If{$\left|U_i\right|=n-f$}
            \State \textbf{output} $R_i$, but continue updating internal sets and sending messages
        \EndIf
    \EndUpon
\EndUpon
\end{algorithmic}
\end{algorithm}

\begin{algorithm}\caption{$\GatherVerify_i(I)$}
    \begin{algorithmic}[1]
        \Upon {$\left|\left\{j|\exists (j,V_j)\in U_i, V_j\subseteq I\right\}\right|\geq n-f\land I\subseteq S_i$}
            \State $X\gets\{(j,x)\in R_i|j\in I\}$
            \State \textbf{output} $X$ and \textbf{terminate}
        \EndUpon
    \end{algorithmic}
\end{algorithm}

\subsection{Security Analysis}

\begin{lemma}\label{lem:weakCore}
    Assume some nonfaulty party completed the protocol.
    There exists some $i^*$ such that at least $f+1$ parties sent broadcasts of the form $\langle 3, T\rangle$ with $i^*\in T$.
\end{lemma}
\begin{proof}
    Assume some nonfaulty party completed the protocol.
    Before completing the protocol, it found that $\left|U_i\right|\geq n-f$, and thus it received $n-f$ broadcasts of the form $\langle 3, T_j\rangle$ such that $\left|T_j\right|\geq n-f$.
    Let $I$ be the set of parties who sent those broadcasts.
    Now assume by way of contradiction that every index $k$ appears in at most $f$ of the broadcasted sets $T_j$ such that $j\in I$.
    Since there are a total of $n$ possible values, this means that the total number of elements in all sets is no greater than $nf$.
    On the other hand, there are $n-f$ such sets, each containing $n-f$ elements or more, resulting in at least $(n-f)^2$ elements overall.
    Combining these two observations:
    \begin{align*}
        (n-f)^2 &\leq nf\\
        n^2-2nf+f^2&\leq nf\\
        n^2-3nf+f^2&\leq 0
    \end{align*}
    However, by assumption $n>3f$, and thus:
    \begin{align*}
        0 &\geq n^2-3nf+f^2 \\
        & = n^2 - n\cdot (3f) +f^2\\
        &> n^2-n^2+f^2\\
        &= f^2\geq 0
    \end{align*}
    reaching a contradiction.
    Therefore, there exists at least one value $i^*$ such that for at least $f+1$ of the $\langle 3, T\rangle$ broadcasts sent, $i^*\in T$.
\end{proof}
\begin{lemma}\label{lem:valueReceived}
    If for some nonfaulty party $i$ $(j,V_j)\in U_i$, then $i$ received a $\langle1,x_k\rangle$ broadcast from every $k\in V_j$ such that $\validate(x_k)=1$.
\end{lemma}
\begin{proof}
    Observe some $(j,V_j)\in U_i$ and $k\in V_j$.
    Before adding $(j,V_j)$ to $U_i$, $i$ saw that $T_j\subseteq T_i$.
    This means that for every $l\in T_j$, $i$ first received a $\langle 2,S_l\rangle$ broadcast from $l$ such that $S_l\subseteq S_i$.
    By definition, $V_j=\bigcup_{l\in T_j} S_l$ and thus $V_j\subseteq S_i$.
    Before adding $k$ to $S_i$, $i$ must have received a $\langle1,x_k\rangle$ validated broadcast checking that $\validate(x_k)=1$, completing the proof.
\end{proof}

\begin{theorem}\label{thm:gather}
    The pair $(\Gather,\GatherVerify)$ is a verifiable reliable gather protocol resilient to $f<\frac{n}{3}$ Byzantine parties.
\end{theorem}
\begin{proof}
    Each property is proven separately.
    
    \textbf{Termination of Output.} 
    Assume that $\validate(x_i)=1$ for every nonfaulty $i$ and that all nonfaulty parties participate in the $\Gather$ protocol.
    The first thing they do is send a $\langle1,x_i\rangle$ message using a validated broadcast.
    By assumption, $\validate(x_i)=1$ for every nonfaulty $i$, and thus
    every nonfaulty $j$ receives the broadcast and updates $R_i$ and $S_i$.
    After receiving a $\langle1,x_j\rangle$ message from every nonfaulty $j$, $\left|S_i\right|=n-f$, so party $i$ sends the message $\langle2,S_i\rangle$.
    Afterwards, every nonfaulty party receives $\langle2,S_j\rangle$ from every nonfaulty $j$.
    Note that since $j$ sent $S_j$, it must have received a $\langle1,x_k\rangle$ validated broadcast from every $k\in S_j$.
    The message was received by validated broadcast, so $i$ eventually receives the same message and adds $k$ to $S_i$ as well.
    Therefore $i$ eventually sees that $S_j\subseteq S_i$ and adds $j$ to $T_i$.
    Finally, after $n-f$ such updates, $i$ broadcasts $T_i$.
    Using similar arguments, every nonfaulty party eventually adds some tuple of the form $(j,V_j)$ to $U_i$ for every nonfaulty $j$.
    Then $i$ sees that $\left|U_i\right|=n-f$ and outputs some value.
    A nonfaulty party $i$ only adds pairs of the form $(j,x)$ to $R_i$ after receiving a validated broadcast of the form $\langle 1, x\rangle$ from party $j$.
    This message was received by validated broadcast, so $\validate(x)=1$, and thus $x\in\mathcal{M}$ as well. 
    Every party can send only one such broadcast, and thus at all times throughout the protocol, $R_i$ consists of pairs $(j,x)$ such that $j\in [n]$ and $x\in\mathcal{M}$ and the index $j$ appears in $R_i$ at most once.
    In other words, $R_i$ is a gather-set throughout the protocol, including when $i$ outputs the set $X=R_i$.
    
    \textbf{Completeness.}
    Assume some nonfaulty party $i$ completes the $\Gather$ protocol and outputs $X_i$.
    Before adding $(k,x_k)$ to $R_i$ and $k$ to $S_i$, party $i$ first receives a $\langle1,x_k\rangle$ validated broadcast from $k$.
    Every nonfaulty $j$ eventually receives the same broadcast and adds $(k,x_k)$ to $R_j$ and $k$ to $S_j$ as well.
    Therefore, eventually $S_i\subseteq S_j$ for every nonfaulty $j$.
    Before adding $k$ to $T_i$, $i$ receives a broadcast $\langle 2, S_k\rangle$ such that $S_k\subseteq S_i$ and $\left|S_k\right|\geq n-f$.
    Since every nonfaulty $j$ eventually receives the same broadcast and $S_i\subseteq S_j$, $j$ also adds $k$ to $T_j$.
    Using similar arguments, before adding $(k,V_k)$ to $U_i$, $i$ receives a broadcast $\langle 3, T_k\rangle$ such that $T_k\subseteq T_i$ and $\left|T_k\right|\geq n-f$.
    Party $j$ eventually receives the same message, sees that the $T_k\subseteq T_i\subseteq T_j$ and $\left|T_k\right|\geq n-f$, and then computes $V_k$ using the exact same $S$ sets $i$ used when computing the set, because all values were received by broadcast.
    Therefore at that point $j$ adds $(k,V_k)$ to $U_j$.
    Now, at the time $i$ outputs a value from the $\Gather$ protocol, it sees that $\left|U_i\right|\geq n-f$, and outputs $R_i$.
    From Lemma~\ref{lem:valueReceived}, at that time for every $(j,V_j)\in U_i$ and $k\in V_j$, $i$ received some $\langle 1,x_k\rangle$ broadcast from party $k$ and thus $k\in S_i$.
    In other words, for every $(j,V_j)\in U_i$, $V_j\subseteq S_i$.
    At all times in the protocol, $Indices(R_i)=S_i$ because an index $k$ is added to $S_i$ at the same time a tuple $(k,x)$ is added to $R_i$.
    This means that if we observe $Indices(X_i)$, which equals $S_i$ at the time $i$ outputs $X_i$, for every $(k,V_k)\in U_i$, $V_k\subseteq S_i=Indices(X_i)$.
    Combining those two observations, every nonfaulty party $j$ eventually sees that for every $(k,V_k)\in U_i\subseteq U_j$, $V_k\subseteq Indices(X_i)$.
    At the time $i$ outputs a value from the $\Gather$ protocol, $\left|U_i\right|\geq n-f$ so there are eventually $n-f$ such tuples in $U_j$ as well.
    Furthermore, $Indices(X_i)=S_i\subseteq S_j$, which means $j$ eventually proceeds to the next line.
    At that time, $j$ computes $X=\{(j,x)\in R_j|j\in Indices(X_i)\}$.
    As stated above, $X_i$ equals $R_i$ at the time $i$ output $X_i$ from the $\Gather$ protocol, and $Indices(X_i)$ equals $S_i$ at that time.
    When $j$ sees that $Indices(X_i)\subseteq S_j$, it has already received a validated broadcast $\langle 1, x_k\rangle$ from every party $k\in Indices(X_i)$ and added $(k,x_k)$ to $R_j$.
    $R_j$ is a gather-set at all times, so this is the same tuple that $j$ added to its output from the $\GatherVerify$ protocol, $X$.
    This is the same broadcast $i$ received, so it added the same tuple $(k,x_k)$ to $R_i$ before outputting $X_i$.
    In other words, $j$ added the same tuple $(k,x_k)$ to $X$ that $i$ added to its output $X_i$.
    Party $j$ only adds tuples of the form $(k,x_k)$ if $k\in Indices(X_i)$, so those are all the tuples in $X$.
    
    \textbf{Agreement on Verification.} Assume that some nonfaulty party $i$ completes protocol $\GatherVerify_i(I)$ on an index-set $I$ and outputs a set $X$, and that all nonfaulty parties participate in the $\Gather$ protocol.
    At the time $i$ completed the protocol, $I\subseteq S_i$ and
    $\left|\left\{k|\exists (k,V_k)\in U_i, V_k\subseteq I\right\}\right|\geq n-f$.
    Let $j$ be some nonfaulty party that runs the protocol $\GatherVerify_j(I)$.
    Before $i$ added some element $(k,x_k)$ to $R_i$ and $k$ to $S_i$, it received a validated broadcast of the message $\langle 1, x_k\rangle$ from $k$.
    From the Termination and Correctness properties of the Validated Reliable Broadcast protocol, $j$ eventually receives that message from $k$ as well and thus $(k,x_k)\in R_j$ and $k\in S_j$ as well.
    In other words, eventually $R_i\subseteq R_j$ and $S_i\subseteq S_j$.
    Before adding an element $k$ to $T_i$, $i$ received a broadcast of a set $S_k$ from $k$ such that $\left|S_k\right|\geq n-f$ and $S_k\subseteq S_i$.
    From the Termination and Correctness properties of the Reliable Broadcast protocol, $j$ eventually receives the same message from $k$.
    As shown above, eventually $S_i\subseteq S_j$, and at that time $j$ adds $k$ to $T_j$ as well.
    Therefore, eventually $T_i\subseteq T_j$.
    Using similar arguments, if there exists some $(k,V_k)$ in $U_i$, then eventually $j$ adds some element $(k,V'_k)$ to $U_j$ as well.
    From the Correctness property of the Reliable Broadcast protocol, $i$ and $j$ receive the same sets $S_l$ from all parties, and thus when computing $V_k$ and $V'_k$, they both do so with the same values.
    This in turn means that they add the same tuple $(k,V_k)$ to their $U_i$ and $U_j$ sets and thus eventually $U_i\subseteq U_j$ as well.
    Combining all of those observations, eventually $I\subseteq S_i\subseteq S_j$.
    In addition, for every $(k,V_k)\in U_i$ such that $V_k\subseteq I$, eventually $(k,V_k)\in U_j$ as well.
    Since there are at least $n-f$ such tuples in $U_i$, there are eventually $n-f$ such tuples in $U_j$ as well.
    When both of those conditions hold, $j$ proceeds to the next line of the $\GatherVerify$ protocol.
    When $i$ completed the protocol, it saw that $I\subseteq S_i$ and thus it received a $\langle 1,x_k\rangle$ from every $k\in I$, and added a tuple $(k,x_k)$ to $R_i$.
    Using the same reasoning, $j$ received broadcasts from the same parties, and from the Agreement property of the validated reliable broadcast protocol, it received the same messages and added the same tuples to $R_j$.
    In other words, $j$ computed $X$ using the same values as $i$, so it output the same set $X$.
    
    \textbf{Agreement.}
    Let $i,j$ be two nonfaulty parties and $I,J$ be two sets such that $\GatherVerify_i(I)$ and $\GatherVerify_j(J)$ eventually terminate with the outputs $X$ and $Y$ respectively.
    Since $\GatherVerify$ terminates in both cases, $I\subseteq S_i, J\subseteq S_j$.
    From the way $i$ calculates $X$ and $j$ calculates $Y$, $X\subseteq R_i$ and $Y\subseteq R_j$.
    Observe a pair of tuples $(k,x)\in X\subseteq R_i,(k,y)\in Y\subseteq R_j$.
    Party $i$ only adds $(k,x)$ to $R_i$ after receiving a broadcast of $\langle 1,x\rangle$ from $k$, and party $j$ adds the tuple $(k,y)$ to $R_j$ after receiving a broadcast of $\langle 1,y\rangle$ from $k$.
    From the Agreement property of the validated reliable broadcast protocol, both $i$ and $j$ received the same broadcast of the form $\langle1, z\rangle$, and thus $x=y$.
    
    \textbf{Binding Core.} Assume the first nonfaulty party that completes the $\Gather$ protocol is $p^*$, and observe the index $i^*$ as defined in Lemma~\ref{lem:weakCore}.
    Party $p^*$ only adds a tuple $(k,V_k)$ to $U_{p^*}$ after receiving a $\langle3,T_k\rangle$ message from party $k$.
    Before completing the protocol, $p^*$ received $n-f$ such broadcasts, and from Lemma~\ref{lem:weakCore}, $f+1$ of the parties broadcast some message $\langle 3,T_k\rangle$ such that $i^*\in T_k$. 
    Therefore for some $(k,V_k)\in U_{p^*}$, $i^*\in T_k$.
    Note that $T_k\subseteq T_{p^*}$, so $i^*\in T_{p^*}$.
    Before adding $i^*$ to $T_{p^*}$, $p^*$ received a $\langle 2,S_{i^*}\rangle$ broadcast from party $i^*$ such that $S_{i^*}\subseteq S_{p^*}$ and $\left|S_{i^*}\right|\geq n-f$.
    Similarly, before adding $k\in S_{i^*}$ to $S_{p^*}$, $p^*$ first receives a $\langle 1,x^*_k\rangle$ broadcast from $k$.
    Let the binding-core $X^*$ be defined as follows: $X^*=\left\{(k,x^*_k)|k\in S_{i^*}\right\}$, i.e. pairs consisting of a party in $S_{i^*}$ and the value that $p^*$ received from that party via broadcast.
    Clearly $\left|X^*\right|\geq n-f$ because $\left|S_{i^*}\right|\geq n-f$.
    The fact that $X^*$ is a subset of every nonfaulty party's output from the protocol is a direct corollary of the Completeness and Includes Core properties of the $\Gather$ protocol. 
    
    \textbf{Internal Validity.} Let $p^*$ be the first nonfaulty party that completed the $\Gather$ protocol, as defined in the Binding Core property.
    Let $j$ be some party that was nonfaulty at that time such that there exists a tuple $(j,x)\in X^*$.
    Let $i^*$ be defined as it is in the Binding Core property and Lemma~\ref{lem:weakCore}.
    By definition, if $(j,x)\in X^*$, then $j$ is in the set $S_{i^*}$ that $p^*$ received from party $i^*$. 
    As shown in the Binding core property, at the time that $p^*$ completed the $\Gather$ protocol, it already received a $\langle 1, x^*_j\rangle$ message from party $j$, and $x$ is defined to be $x^*_j$.
    Now, since $j$ was nonfaulty at that time, it broadcasted the message $\langle 1, x_j\rangle$, with $x_j$ being its input to the protocol.
    Therefore, $x=x_j$ as required.
    
    \textbf{Include Core.} Let $i$ be some nonfaulty party and $I$ be some index set such that $\GatherVerify_i(I)$ terminates with the output $X$.
    Party $i$ found that $\left|\left\{k|\exists (k,V_k)\in U_j, V_k\subseteq I\right\}\right|\geq n-f$.
    As discussed above, party $i$ only adds $(j,V_j)$ to $U_i$ after receiving a $\langle 3,T_j\rangle$ message from $j$.
    Let $i^*$ be defined as it is in Lemma~\ref{lem:weakCore} and in the Binding Core property.
    Seeing as there are at least $f+1$ parties that sent broadcasts of the form $\langle 3,T\rangle$ with $i^*\in T$ and $n-f$ parties $j$ such that $(j,V_j)\in U_i$ and $V_j\subseteq I$, for at least one of those parties $i^*\in T_j$.
    By definition, $V_j=\bigcup_{k\in T_j} S_k$, and thus $S_{i^*}\subseteq V_j\subseteq I$.
    Therefore, for every $k\in S_{i^*}\subseteq I$, party $i$ adds a tuple $(k,x)$ to its output $X$.
    Finally, $X\subseteq R_i$, and $i$ only adds $(j,x)$ to $R_i$ after receiving a $\langle 1,x\rangle$ broadcast from $j$.
    Let $p^*$ be the first nonfaulty party that completed the $\Gather$ protocol as defined in the Binding Core property.
    Since $k\in S_{i^*}$, $p^*$ received a $\langle 1, x^*_k\rangle$ broadcast from $k$, so it must be the case that $x=x^*_k$ as defined in the Binding Core Property.
    In other words, for every $k\in S_{i^*}$, $(k,x^*_k)\in X$, and thus $X^*\subseteq X$.
    
    \textbf{External Validity.} Assume that for some nonfaulty $i$, $\GatherVerify_i(I)$ terminates.
    When $i$ completed protocol it outputs $\{(j,x)\in R_i|j\in I\}\subseteq R_i$
    Party $i$ adds $(j,x)$ to $R_i$ only after receiving a validated broadcast of $\langle1,x\rangle$ from $j$ checking that $\validate(x)=1$.
\end{proof}

\section{Proposal Election} \label{sec:weakleader}
In this section we construct a verifiable weak proposal election, which is related to the idea of a weak common coin.
With constant probability all nonfaulty parties output the proposal of a nonfaulty party, but in other cases parties might output different values.
The protocol is also externally validated, meaning that every party's output is externally valid.
In addition, the protocol is verifiable.
Like in the case of the Verifiable Gather protocol, this means that parties can prove to each other that the value they output is indeed a viable output from the protocol.
In the case that a single nonfaulty party's input is chosen, this means that this is the only value that will pass verification.
Our construction uses techniques inspired by Katz and Koo's synchronous weak leader election \cite{KatzK06}. 
They use verifiable secret sharing in order to determine the leader through a random coin whose value can only be obtained at the end of the protocol i.e. after reconstruction. 
We extend their results to the asynchronous setting by making use of a threshold verifiable random function (VRF) instantiated using a (local) DKG.
There is a VRF public key associated to every player, and this public key is entirely determined by that player (provided it contains sufficient secret key shares).  Parties cannot trivially reach consensus about a single DKG because they do not know if there are DKG transcripts that have been received by other parties, but not by them.

The protocol proceeds in four rounds and pseudocode is provided in \cref{alg:pe}.
In the first round, every party sends a VRF-DKG share to every other party.
If some party wishes their proposal to be considered it must input a pair consisting of their proposal and an aggregated VRF-DKG transcript into the $\Gather$ protocol.
This essentially forces parties to commit to those values because only one tuple of the form $\gatherTuple{j}$ may appear in any of the outputs from $\Gather$ for any given $j$.
After outputting the gather set $X$ from the $\Gather$ protocol, every party broadcasts $Indices(X)$, which is the set of indices with tuples in $X$.
After receiving an index-set for which $\GatherVerify$ terminates with the output $X$, parties send VRF evaluation shares for all tuples in $X$, if they haven't done so earlier.
Note that at this time all of the tuples in $X$ have already been committed to because of the Agreement property of the $\Gather$ protocol.
After receiving $n-f$ evaluation shares for each of the tuples in the output from the $\Gather$ protocol, every party evaluates the VRF at the appropriate values, and chooses the proposal with the highest corresponding VRF evaluation. 
We think of the PE protocol as succeeding if the maximal evaluation corresponds to a tuple in the binding core that corresponds to a value input by a nonfaulty party.
As will be shown below, this happens with a constant probability, and when that happens all parties output the corresponding proposal.

The protocol proceeds in a few conceptual rounds described below:
\\

\noindent  \textbf{Round 1:}
In Round $1$, each party samples and sends a VRF-DKG share for every other party.
The VRF will later be used to assign a number to each party.
Party $i$ waits to receive $n-f$ valid contributions from all other parties.
It then aggregates these VRF-DKG contributions into a verifying VRF-DKG transcript $\agg_i$.
\\

\noindent  \textbf{Round 2:}
In Round $2$ party $i$ calls the $\Gather$ protocol providing its original input $\proposal_i$ and the aggregated VRF-DKG transcript $\agg_i$ as input.
From the properties of the $\Gather$ protocol, each party will eventually output a set of tuples $\gatherTuple{j}$ indicating that $j$ input the pair $\proposal_j$ and $\agg_j$ to the protocol.
\\

\noindent  \textbf{Round 3:}
After outputting a gather-set from the $\Gather$ protocol, parties can start calculating the number assigned to each party.
Ideally, each party would send the gather-set they output from the protocol to all other parties, and they will help in evaluating all of the relevant values.
However, having another all-to-all communication round where parties send sets of $O(n)$ tuples containing $O(m)$ words each would incur an overhead of $O(mn^3)$ words to be sent.
Instead of doing that, every party only broadcasts the indices of tuples in its gather-set, which we think of as a request to start evaluating the VRF for each index.
\\

\noindent  \textbf{Round 4:}
After receiving an index-set $I$, every nonfaulty party calls the $\GatherVerify$ protocol on the set, and waits to output the tuples corresponding to those indices.
After that happens parties send their evaluation share for each tuple they haven't seen yet.
This is done by maintaining a set $\startEval$ which stores all of the seen tuples.
When a party completes the $\GatherVerify$ protocol with the output $X$, it first sends an evaluation share for every tuple in $X\setminus \startEval$, and only then updates $\startEval$ to contain $X$. 

Crucially, the proposal and aggregated VRF-DKG transcript are sent together, and parties start sending the VRF evaluation shares only after seeing the relevant aggregated VRF-DKG transcript included in a gather-set received as output from the $\GatherVerify$.
By sending the proposal and VRF-DKG transcript together, parties have to commit to their values before knowing which party's proposal is going to "win" the election.
From the properties of the $\Gather$ protocol, once a tuple $\gatherTuple{j}$ is in a gather-set output from $\GatherVerify$, no other party ever outputs a gather-set from $\GatherVerify$ with a different tuple corresponding to the index $j$.
By sending evaluation shares only then, nonfaulty parties guarantee that the faulty parties committed to their aggregated VRF-DKG transcript before knowing what number it evaluates to.
This guarantees that those evaluations cannot be biased by the faulty parties.

After receiving enough evaluation shares to compute $\phi(\agg_j,\langle j\rangle)$ for every $\gatherTuple{j}$ in their output from the $\Gather$ protocol, party $i$ chooses the index $\ell$ with the maximal value $\phi(\agg_\ell,\langle \ell\rangle)$ and outputs $\proposal_\ell$. In addition, $i$ outputs the indices of parties in their gather-set as proof.
\\

Intuitively, every party outputs a gather-set from the $\Gather$ protocol which determines the VRF evaluations taken into consideration.
If the VRF evaluation with the maximal value among all outputs from the $\Gather$ protocol corresponds to a tuple $(\ell^*,(\proposal_{\ell^*},\agg_{\ell^*}))$ in the binding core of the $\Gather$ protocol that was input by a nonfaulty party, then all nonfaulty parties will see that evaluation and pick $\proposal_{\ell^*}$ as their output.
Since the evaluations are sampled uniformly in an unbiased manner, this means that every party has the same probability of having the maximal evaluation being associated with it.
When counting the number of nonfaulty parties with tuples in the common core, we find that the probability of the aforementioned event is at least $\frac{1}{3}$.
This mechanism also allows to check whether a given proposal could have been the correct output from the $\PE$ protocol.
In order to convince a nonfaulty party that a value is a correct output from the $\PE$ protocol, it is enough to provide one's output from the $\Gather$ protocol.
Parties will then be able to check if that is a verifying gather-set and if the correct proposal was elected based on that output.
Instead of actually using the whole gather-set as proof, only the indices of tuples in it are sent as proof in order to reduce communication.
If the maximal evaluation is associated with a tuple in the binding-core, then only gather-sets containing that tuple will verify, which means that only $\proposal_{\ell^*}$ as defined above will verify.
\\

\noindent \textbf{Verification:}
The verification algorithm is given in \cref{alg:pe:verification}.
As stated above, in order for a value $x$ to verify with a proof $\pi$, parties require the indices of the gather-set with which it was computed.
They then check if the index-set verifies, if all the relevant tuples have been previously received, and if the evaluation of the VRF has been computed at all relevant points.
If all of those conditions hold, parties then make sure that $x$ is the proposal with the maximal associated VRF evaluation.

\begin{algorithm}\caption{$\PE_i(\proposal_i)$}\label{alg:pe}
\begin{algorithmic}[1]
\State $\dkgShares_i\gets\emptyset, 
X_i\gets \emptyset, \forall j\in[n]\  \evalShares_i[j] \gets \emptyset, evals_i\gets\emptyset, \startEval_i\gets\emptyset$
\State $(\share_{i,1}, \ldots, \share_{i,n}) \randpick \generateShare(\sk_i), \ldots, \generateShare(\sk_i)$ 
\State for every $j\in[n]$ send $\langle \dkg, \share_{i,j} \rangle $ to $j$

\Upon{receiving the first $\langle \dkg, \share_{j,i} \rangle$ from $j$ message such that $\verifySecret(\pk_j, \share_{j,i}) = 1$}
    \State $\dkgShares_i\gets \dkgShares_i\cup \{ \share_{j,i} \}$
        \If{$\left|\dkgShares_i\right|=n-f$}
            \State $\agg_i \gets \aggregateShares(\dkgShares_i)$
            \State \textbf{call} $ \Gather_i(\proposal_i,\agg_i)$ with the external validity function $\checkValidity$
        \EndIf
\EndUpon

\Upon{$\Gather_i$ outputting the set $X=\{\gatherTuple{j}\}$} \Comment{continue updating state according to $\Gather$}
    \State $X_i\gets X$
    \State $I_i\gets Indices(X_i)=\{k|\exists \gatherTuple{k}\in X_i\}$
    \State \textbf{broadcast} $\langle indices, I_i\rangle$
\EndUpon

\Upon{receiving the first $\langle indices, I_j\rangle$ message from $j$}
    \Upon{$\GatherVerify_i(I_j)$ terminating with output $X_j$ and $\Gather$ outputting some value}
        \ForAll{$(k,\proposal_k,\agg_k)\in X_j\setminus \startEval_i$}
            \State $(\evalShare_{k,i},\pi_{k,i})\gets \Eval\Share(\agg_k,\sk_i,\langle k\rangle)$
            \State send $\langle eval, k, \evalShare_{k,i},\pi_{k,i} \rangle$ to every party
        \EndFor
        \State $\startEval_i\gets \startEval_i\cup X_j$
    \EndUpon
\EndUpon

\Upon{receiving the first $\langle eval, k, \evalShare_{k,j},\pi_{k,j} \rangle$ broadcast from $j$ for any given $k$}
    \Upon{$\exists \gatherTuple{k}\in \startEval_i$}
        \If{$\EvalShareVerify(\agg_k,\pk_j,\langle k\rangle, \evalShare_{k,j},\pi_{k,j})=1$}
            \State $\evalShares_i[k]\gets \evalShares_i[k]\cup\{(\evalShare_{k,j},\pi_{k,j})\}$
            \If{$\left|\evalShares_i[k]\right|=n-f$}
                \State $(\evaluation_k,\pi_k)\gets \Eval(\agg_k,\langle k\rangle,\evalShares_i[k])$
                \State $\evals_i\gets  \evals_i\cup\{(k,\evaluation_k)\}$
            \EndIf
        \EndIf
    \EndUpon
\EndUpon
\Upon{$\forall (k,(\agg_k,\proposal_k))\in X_i\ \exists (k,\evaluation_k)\in \evals_i$ and $X_i\neq\emptyset$}
    \State $\ell\gets argmax_{k} \{\evaluation_k|\gatherTuple{k}\in S_i\}$ \Comment{i.e. $\ell$ has the maximal $\evaluation_\ell$}
    \State $\pi_i\gets Indices(X_i)=\{k|\exists \gatherTuple{k}\in X_i\}$
    \State \textbf{output} $(\proposal_\ell,\pi_i)$, but continue updating internal sets and sending messages
\EndUpon
\end{algorithmic}
\end{algorithm}

\begin{algorithm}\caption{$\checkValidity(\proposal,\agg)$}
\begin{algorithmic}[1]
    \If{$\validate(\proposal)=1$ and $\DKG\Verify(\agg)=1$}
        \State \textbf{return} 1
    \Else
        \State \textbf{return} 0
    \EndIf
\end{algorithmic}
\end{algorithm}

\begin{algorithm}\caption{$\PEVerify_i(x,\pi)$} \label{alg:pe:verification}
    \begin{algorithmic}[1]
        \Upon{$\forall k\in\pi\ \exists (k,\evaluation_k)\in \evals_i\land \exists \gatherTuple{k}\in \startEval_i$}
            \Upon{$\GatherVerify_i(\pi)$ terminating}
                \State $\ell\gets argmax_{k} \{\evaluation_k|k\in\pi\}$
                \If{$x=\proposal_\ell$}
                    \State \textbf{terminate}
                \EndIf
            \EndUpon
        \EndUpon
    \end{algorithmic}
\end{algorithm}

\subsection{Security Analysis}

The following lemmas show that the $\startEval$ and $\evals$ sets of different parties are eventually consistent with each other.
\begin{lemma}\label{lem:startEval}
    If all nonfaulty parties participate in the $\PE$ protocol, and some nonfaulty party $i$ outputs the set $X_i$ from the $\Gather$ protocol, then for every nonfaulty $j$ eventually $X_i\subseteq \startEval_j$.
    Furthermore, if for two nonfaulty parties $i,j$, $\gatherTuple{k}\in \startEval_i$ and $(k,(\proposal'_k,\agg'_k))\in \startEval_j$, then $(\proposal_k,\agg_k)=(\proposal'_k,\agg'_k)$.
\end{lemma}
\begin{proof}
    If some nonfaulty party output $X_i$ from the gather protocol, then it broadcasts $\langle indices, Indices(X_i)\rangle$.
    Every nonfaulty $j$ receives that message, calls $\GatherVerify(Indices(X_i))$ and from the Completeness property of the $\Gather$ protocol, eventually outputs $X_i$.
    After that time, $j$ performs some local computations and updates $\startEval_j$ to $\startEval_j\cup X_i$.
    
    Now observe two nonfaulty parties $i,j$ such that $\gatherTuple{k}\in \startEval_i$ and $(k,(\proposal'_k,\agg'_k))\in \startEval_j$.
    Before adding $(k,(\proposal_k, \agg_k))$ to $\startEval_i$, $i$ output some set $X$ from $\GatherVerify$ with $(k,(\proposal_k,\\ \agg_k))\in X$.
    Similarly, before adding $(k,(\proposal'_k,\agg'_k))$ to $\startEval_j$, $i$ output some set $Y$ from $\GatherVerify$ with $(k,(\proposal'_k,\agg'_k))\in Y$.
    Therefore, $(\proposal_k,\agg_k)=(\proposal'_k, \agg'_k)$ from the Agreement property of the $\Gather$ protocol.
\end{proof}

\begin{lemma}\label{lem:evalConsistency}
    If $\gatherTuple{k}\in\startEval_i$ for some nonfaulty $i$, then eventually for every nonfaulty $j$, there exists a tuple $(k,\phi(\agg_k,\langle k\rangle))\in \evals_j$.
    Furthermore, if $(k,\evaluation_k)\in \evals_i$ for some nonfaulty $i$, then there exists some tuple $\gatherTuple{k}\in\startEval_i$ such that $\evaluation_k=\phi(\agg_k,\langle k\rangle)$.
\end{lemma}
\begin{proof}
    If $\gatherTuple{k}\in \startEval_i$, then $i$ added that tuple after receiving some broadcast $\langle indices, I\rangle$ for which $\GatherVerify_i(I)$ terminated with an output $X$ such that $\gatherTuple{k}\in X$.
    From the Termination and Agreement properties of the broadcast protocol, every other nonfaulty $j$ eventually receives the same message.
    From the Agreement on Verification property of the $\Gather$ protocol, eventually $j$ outputs the same $X$ from $\GatherVerify_j(I)$, and then adds $\gatherTuple{k}$ to $\startEval_j$.
    A tuple $\gatherTuple{k}$ is added to $\startEval_j$ only after already sending $\langle eval, k, \evalShare_{k,j},\pi_{k,j}\rangle$, so all nonfaulty parties send such a message for every $\gatherTuple{k}\in \startEval_i$.
    Therefore, for every $\gatherTuple{k}\in \startEval_i$, every nonfaulty party $j$ receives a a message $\langle eval, l, \evalShare_{k,l},\pi_{k,l}\rangle$ from every nonfaulty $l$, and sees that $\gatherTuple{k}\in \startEval_j$.
    Since a nonfaulty $l$ computed the share correctly, $\EvalShareVerify(\agg_k,\pk_l,\langle k\rangle,\evalShare_{k,l},\pi_{k,l})=1$.
    Party $j$ then adds the tuple $(\evalShare_{k,l},\pi_{k,l})$ to $\evalShares_j[k]$.
    After adding such a tuple for every nonfaulty party, $j$ sees that $\left|\evalShares_j[k]\right|=n-f$, it computes $\phi(\agg_k,\langle k\rangle),\pi_k$ using $\Eval$ and adds the tuple $(k,\phi(\agg_k,\langle k\rangle))$ to $\evals_j$.
    
    Now, let $(k,\evaluation_k)\in \evals_i$ for some nonfaulty $i$.
    Before adding that tuple to $\evals_i$, party $i$ saw that $\exists\gatherTuple{k}\in \startEval_i$ and added $n-f$ shares to $\evalShares_i[k]$.
    It then computed $(\evaluation_k,\pi_k)=\Eval(\agg_k,\langle k\rangle,\evalShare_i[k])$ and added $(k,\evaluation_k)$ to $\evals_i$.
    From the definition of the VRF, $\evaluation_k=\phi(\agg_k,\langle k\rangle)$.
\end{proof}

\begin{corollary}\label{col:sameEvals}
    Let $i,j$ be two nonfaulty parties such that $(k,\evaluation_k)\in \evals_i$ and $(k,\evaluation'_k)\in \evals_j$.
    Then $\evaluation_k=\evaluation'_k$.
\end{corollary}
\begin{proof}
    From \cref{lem:evalConsistency}, there exists a tuple $\gatherTuple{k}\in \startEval_i$ such that $\evaluation_k=\phi(\agg_k,\\ \langle k\rangle)$.
    Similarly, there exists a tuple $(k,(\proposal'_k,\agg'_k))\in \startEval_j$ such that $\evaluation'_k=\phi(\agg'_k,\langle k\rangle)$.
    From \cref{lem:startEval}, $(\proposal_k,\agg_k)=(\proposal'_k,\agg'_k)$, so $\evaluation_k=\evaluation'_k$.
\end{proof}

\begin{theorem}\label{thm:vwpe}
    The pair $(\PE,\PEVerify)$ is a verifiable weak proposal election protocol resilient to $f<\frac{n}{3}$ parties with $\alpha=\frac{1}{3}$.
\end{theorem}
\begin{proof}
    Each property is proven separately.
    
    \textbf{Termination of Output.}
    If all nonfaulty parties participate in the $\PE$ protocol, then they all send a $\langle dkg, \share_{i,j} \rangle $ message to every other party, with $\share_{i,j}$ being generated using $\generateShare$. 
    Every nonfaulty party $i$ eventually receives at least $n-f$ shares from the nonfaulty parties such that $\verifySecret(\pk_i, \share_{j,i}) = 1$ and adds $\share_{j,i}$ to $\dkgShares_i$.
    After that, $i$ sees that $\left| \dkgShares_i\right|=n-f$, it aggregates those shares into $\agg_i$, and inputs $(\proposal_i,\agg_i)$ to the $\Gather$ protocol.
    From the Correctness property of the DKG, $\DKG\Verify(\agg_i)=1$, because $\agg_i$ is an aggregation of $n-f$ verifying DKG shares.
    By assumption, all nonfaulty parties have externally valid inputs (i.e. for every nonfaulty $i$, $\validate(\proposal_i)=1$), so for every nonfaulty $i$ $\checkValidity(\proposal_i,\agg_i)=1$.
    By the Termination of Output property of the $\Gather$ protocol, every nonfaulty party $i$ eventually outputs some set $X_i$ from the protocol.
    From \cref{lem:startEval}, every nonfaulty party $j$ eventually has $X_i\subseteq \startEval_j$.
    In addition, from \cref{lem:evalConsistency}, for every $\gatherTuple{k}\in X_i\subseteq \startEval_i$ eventually there exists a tuple $(k,\evaluation_k)\in \evals_i$.
    At that point, $i$ preforms some local computations and outputs a value from the protocol.
    
    \textbf{Completeness.}
    Assume some nonfaulty party $i$ outputs the value $x$ and proof $\pi$ from $\PE$.
    The way $i$ computes $\pi$ is by taking its output from the $\Gather$ protocol, $X_i$, and computing $\pi=Indices(X_i)$.
    Observe some nonfaulty party $j$ that calls $\PEVerify_j(x,\pi)$.
    From \cref{lem:startEval}, eventually $X_i\subseteq \startEval_j$, so for every $ k\in\pi=Indices(X_i)$ there exists some tuple $\gatherTuple{k}\in \startEval_j$.
    From \cref{lem:evalConsistency}, eventually for every such $k$, there also exists a tuple $(k,\evaluation_k)\in\evals_j$.
    Therefore eventually $j$ proceeds past the first condition of $\PEVerify$.
    Afterwards, $j$ calls $\GatherVerify_j(\pi)$.
    By definition $\pi=Indices(X_i)$, so $\GatherVerify_j(\pi)$ eventually terminates because of the Completeness property of the $\Gather$ protocol.
    Before terminating, $i$ also saw that for every $k\in\pi$ there existed a tuple $(k,\evaluation_k)\in\evals_i$.
    It then computed the index $\ell$ with the maximal $\evaluation_\ell$ and output $\proposal_\ell$.
    From \cref{col:sameEvals}, $j$ has the same tuples $(k,\evaluation_k)\in \evals_j$ so it computes the same $\ell$.
    Similarly, from \cref{lem:startEval}, when $j$ checks if $x=\proposal_\ell$ it does so with the tuple $\gatherTuple{k}\in X_i\subseteq\startEval_j$, and thus from the way $i$ computes $x$, $j$ sees that $x$ is indeed $\proposal_\ell$.
    Note that \cref{lem:startEval} and \cref{col:sameEvals} also imply that the $\startEval$ and $\evals$ sets have only one tuple of the form $\gatherTuple{k}$ for any given $k$, meaning that the values above are unique and well-defined. 
    
    $\alpha$-\textbf{Binding.}
    At the time the first nonfaulty party completes the $\Gather$ protocol, there exists a binding-set $X^*$ of tuples $\gatherTuple{j}$ that must be included in any output of the $\GatherVerify$ protocol.
    Now, observe all of the sets $X$ which are the output of $\GatherVerify_i$ for any nonfaulty $i$ throughout the rest of the protocol, and let $outputs=\bigcup X$ be the set of all tuples $\gatherTuple{j}$ in those sets.
    From the Agreement property of the $\Gather$ protocol, for any given $j\in[n]$ there can be no more than one such tuple $\gatherTuple{j}\in outputs$.
    Furthermore, from the External Validity property of the $\Gather$ protocol, $\checkValidity(\proposal_j,\agg_j)=1$ for every such $j$, and thus $\DKG\Verify(\agg_j)=1$.
    In other words, every such $\agg_j$ is an aggregation of correct shares from at least $f+1$ different parties, and at least one of those parties is nonfaulty.
    
    Since each aggregated VRF-DKG transcript $\agg_j$ contains shares from at least one nonfaulty party, before some nonfaulty party sends its evaluation share of $\agg_j$, the value $\phi(\agg_j,\langle j\rangle)$ is distributed uniformly and independently from the view of the adversary or any single nonfaulty party.
    That is true because of the Unbiasability property of the threshold verifiable random function.
    No nonfaulty party $i$ sends its evaluation share of any of the aggregated VRF-DKGs $\agg_j$ (or their respective non-aggregated shares) before completing the $\GatherVerify$ protocol and outputting a set $X$ from $\GatherVerify_i$ such that $\gatherTuple{k}\in X$.
    At that point, $\gatherTuple{k}$ is already set and every nonfaulty party that outputs a set $X$ from $\GatherVerify$ that contains a tuple with the index $k$, does so with the tuple $\gatherTuple{k}$.
    Combining the fact that no nonfaulty party sends an evaluation share for $\gatherTuple{k}$ before outputting a gather-set containing it from $\GatherVerify$, and that before that happens the value is distributed uniformly and independently from the adversary's view, $\phi(\agg_k,\langle k\rangle)$ is distributed uniformly and independently for every $\gatherTuple{k}\in outputs$.
    In particular, each one of those values has the same probability of being the maximal one, regardless of the adversary's actions.
    
    Now, if $\ell^*=argmax_j\{\phi(\agg_j,\langle j\rangle)|\gatherTuple{j}\in outputs\}$ for some $(\ell^*,(\proposal_{\ell^*},\agg_{\ell^*}))\in X^*$, and party $\ell^*$ is nonfaulty at the time the first nonfaulty party completes the $\Gather$ protocol, define $x^*$ to be $\proposal_{\ell^*}$, otherwise define $x^*=\perp$.
    Note that $X^*$ is at least of size $n-f$, so at least $n-2f$ of the parties $j$ such that there exists a tuple $\gatherTuple{j}\in X^*$ are nonfaulty at the time the first nonfaulty party completes the $\Gather$ protocol.
    From the Internal Validity property of the $\Gather$ protocol, for any party $j$ that was nonfaulty at the time the first nonfaulty party completed the $\Gather$ protocol, the tuple $\gatherTuple{j}$ includes the values $\proposal_j$ and $\agg_j$ that $j$ input to the protocol.
    Each one of those parties has a $\frac{1}{n}$ probability of having the maximal value, and thus the probability that $x^*$ is the input of one of the parties that was nonfaulty at that time is at least $\frac{n-2f}{n}\geq (\frac{n}{3}+1)\cdot\frac{1}{n}=\frac{1}{3}+\frac{1}{n}$.
    Clearly, since they are nonfaulty at that time, they must have also acted in a nonfaulty manner when starting the $\PE$ protocol.
    This analysis ignores the probability of two parties having the same maximal value.
    The probability of this event can be bounded by $\frac{n^2}{2^\lambda}$ since there are $2^\lambda$ different possible values for outputs of $\phi$.
    For the probability to remain at least $\frac{1}{3}$ even when taking the possibility of a collision into consideration, it is enough that the security parameter is at least $3\log(n)$.
    
    \textbf{Agreement on Verification} Let $i,j$ be two nonfaulty parties and $x,\pi$ be two values such that $\PEVerify_i(x,\pi)$ terminates.
    The first thing $i$ does in $\PEVerify$ is wait until $\forall k\in \pi$, there exists a tuple $(k,\evaluation_k)\in \evals_i$ and a tuple $\gatherTuple{k}\in \startEval_i$.
    Party $i$ only updated its $\startEval_i$ set after receiving a broadcast of the form $\langle indices, I\rangle$ and seeing that $\GatherVerify_i(I)$ terminates and outputs the set $X$.
    When that happens, $i$ updates $\startEval_i$ to be $\startEval_i\cup X$.
    From the Termination and Agreement properties of the broadcast protocol, $j$ eventually receives the same message.
    It then runs $\GatherVerify_j(I)$ and eventually outputs the same set $X$ because of the Agreement on Verification property of $\Gather$.
    Afterwards, it also updates $\startEval_j$ to be $\startEval_j\cup X$.
    In other words, for every $\gatherTuple{k}\in \startEval_i$, eventually $\gatherTuple{k}\in \startEval_j$ as well.
    From \cref{lem:evalConsistency}, eventually for every $k\in\pi$ there also exists a tuple $(k,\evaluation'_k)\in \evals_j$.
    Recall that there also exists a tuple $(k,\evaluation_k)\in \evals_i$, and $\evaluation_k=\evaluation'_k$ because of \cref{col:sameEvals}.
    By \cref{lem:startEval} and \cref{col:sameEvals}, $i$ and $j$ only have one such tuple in their respective $\startEval$ and $\evals$ sets, and thus all of the calculations in the rest of the protocol are well defined.
    Before terminating, $i$ called $\GatherVerify_i(\pi)$, which eventually terminated.
    From the Agreement on Verification property of the $\Gather$ protocol, $\GatherVerify_j(\pi)$ also eventually terminates.
    Afterwards, $i$ and $j$ perform the same deterministic non-interactive computation which only depends on the values in $\evals$ and $\startEval$.
    We've shown that $i$ and $j$ have the same values in the relevant tuples, so since $i$ eventually completed the $\PEVerify$ protocol, so does $j$.
    
    \textbf{Binding Verification.}
    If $x^*$ as defined in the $\alpha$-Binding property equals $\perp$, the property trivially holds.
    Assume that $x^*\neq\perp$ and that $\PEVerify_i(x,\pi)$ terminates for some nonfaulty $i$.
    Before $\PEVerify$ terminated, $i$ checked that for every $k\in\pi$ there exists a tuple $(k,\evaluation_k)\in \evals_i$ and a tuple $\gatherTuple{k}\in \startEval_i$.
    From \cref{lem:evalConsistency} if $(k,\evaluation_k)\in \evals_i$ then there exists a tuple $\gatherTuple{k}\in \startEval_i$ such that $\phi(\agg_k,\langle k\rangle)=\evaluation_k$, and from \cref{col:sameEvals} there is only one tuple with the index $k$ in $\evals_i$.
    Combining these observations, for every $k\in\pi$, there exists a tuple $\gatherTuple{k}\in\startEval_i$ and a tuple $(k,\phi(\agg_k,\langle k\rangle))\in\evals_i$ (and no other tuple with the index $k$).
    
    Afterwards, $i$ calls $\GatherVerify_i(\pi)$, which eventually terminates with an output $X$ such that $Indices(X)=\pi$.
    In addition, from the Includes Core property of the $\Gather$ protocol, $X^*\subseteq X$, and thus $Indices(X^*)\subseteq Indices(X)=\pi$.
    Now, note that $i$ only adds a tuple $\gatherTuple{k}$ to $\startEval_i$ if it outputs a gather-set from $\GatherVerify$ that includes $\gatherTuple{k}$, and thus $\startEval_i\subseteq outputs$.
    By definition, $\ell^*$ is the index with the maximal evaluation $\phi(\agg_k,\langle k\rangle)$ among all tuples $\gatherTuple{k}\in outputs$.
    Also, by definition, $\ell^*\in Indices(X^*)\subseteq \pi$.
    Therefore, when $i$ computes $\ell= argmax_{k} \{\evaluation_k|k\in\pi\}$, it sees that the index corresponding to the maximal such value must be $\ell^*$, and so it checks that $x=\proposal_{\ell^*}$ for the tuple $\gatherTuple{\ell^*}\in \startEval_i$.
    As discussed above, this is the same $\gatherTuple{\ell^*}$ tuple in $outputs$, so $\proposal_{\ell^*}=x^*$.
    Party $i$ eventually terminated, and thus it found that $x=\proposal_{\ell^*}=x^*$, as required.
    
    \textbf{External Validity.}
    Observe some nonfaulty party $i$, value $x$ and proof $\pi$ such that $\PEVerify_i(x,\pi)$ terminates.
    Since $\PEVerify$ terminates, $i$ must have found that $x=\proposal_\ell$ for some $\gatherTuple{\ell}\in \startEval_i$.
    Party $i$ only updates $\startEval_i$ by adding all elements in $X_j$ after $\GatherVerify_i$ outputs the set $X_j$.
    From the External Validity property of the $\Gather$ protocol, for every $\gatherTuple{k}\in X_j$, $\checkValidity(\proposal_k,\agg_k)=1$, which in turn means that $\validate(\proposal_k)=1$.
    This is true for $\proposal_\ell$ as well.
\end{proof}

\section{No Waitin' HotStuff}\label{sec:consensus}
We present a new primary-backup based consensus protocol for the asynchronous model: \textit{No Waitin' Hotstuff} ($\AsyncHS$). As the name suggests, many of the techniques and inspiration for this protocol originated in HotStuff~\cite{HS19}.  Unlike basic HotStuff which requires eventual synchrony, $\AsyncHS$ obtains liveness using the $\PE$ protocol described in Section~\ref{sec:weakleader}, and thus avoids depending on a leader.  The purpose of $\AsyncHS$ is to determine whether or not the $\PE$ protocol was successful, and if not to allow parties to repeat the $\PE$ until consensus is reached. Recall that with probability $\alpha$ (in this implementation $\alpha=\frac{1}{3}$), all parties output the input of a party that was nonfaulty when starting $\PE$.  On the other hand, with probability $1-\alpha$, the parties might output the value that a faulty party input, or even different values from different parties.  Using $\AsyncHS$ we can amplify our constant probability of agreement to an overwhelming probability of agreement.

$\AsyncHS$ proceeds in virtual rounds called ``views'', which are attempts to achieve consensus on the output of the $\PE$ protocol.  
$\AsyncHS$ uses a ``Key-Lock-Commit'' paradigm that helps maintain safety and liveness.
\begin{itemize}[label={}]
    \item \textbf{Key:}  Parties set a local key field that indicates that no other value was committed to in previous rounds.  The keys help maintain liveness: if at any point some party sets a lock in a view where no commitment takes place, then they will eventually see a key from that view (or a later view), that will convince them to participate in the current view.
    
    A key consists of three values: $key$, which is a view number, $key\_val$ which is a value and $\pi$, which is a proof that the key was set correctly in that view.
    
    \item \textbf{Lock:} Before committing to a value in a given view, parties will wait to hear that enough other parties have set a lock on the same value in that view.  
    Before parties set a lock in a given view, they make sure that enough other parties have set a local key field that indicates that no other value was committed to in previous rounds.
    Parties that are locked on a value won't be willing to participate in any later view with a different value.
    They will ignore the lock if and only if enough proof, in the form of a key from a later view, is provided that no commitment actually took place in the view where the lock was set.
    This mechanism helps in guaranteeing the safety of decision values.
    If a commitment took place, then there will be a large number of nonfaulty parties that are locked on that value.
    Those parties won't be willing to participate in views with different values, which will prevent any party from setting a key in a later view with a different value.
    This in turn will guarantee that no party will be able to provide erroneous proof that the locks can be opened.
    
    A lock looks much like a key and consists of three values: $lock$, which is a view number, $lock\_val$ which is a value and $\pi$, which is a proof that the lock was set correctly in that view.
    
    \item \textbf{Commit:} If a nonfaulty party commits to a value no other nonfaulty party ever commits to another value.  
    The locking mechanism guarantees that nonfaulty parties cannot commit to different values.
    In order to help other parties terminate, nonfaulty parties send commit messages to all other parties with proof that the commitment is correct and that they can terminate and output the same value.
\end{itemize}

Algorithm~\ref{alg:asynchs} formally describes $\AsyncHS$.  
It relies on three protocols: $\viewChange$ (Algorithm~\ref{alg:viewChange}) for the first round of interaction in each view and the $\PE$ protocol, and on $\processMessages$ (Algorithm~\ref{alg:processMessages}) and $\processErrors$ (Algorithm~\ref{alg:faults}) for all subsequent rounds in each view.  

Almost all the work takes place in $\processMessages$ (Algorithm $\ref{alg:processMessages}$), in which parties process $echo$, $key$ and $lock$ messages.
Algorithms~\ref{alg:termination} and  \ref{alg:faults} are utilities for processing $commit$, $blame$ and $equivocate$ messages if they are received and either terminating or continuing to the next view if needed.

Finally, the algorithms for checking that $key$, $lock$ and $commit$ messages are correct are provided in Algorithms~\ref{alg:keyCorrect}, \ref{alg:lockCorrect} and $\ref{alg:commitCorrect}$ respectively.
This is done by checking that the provided proof contains signatures from $n-f$ parties on a message from the previous round. For example a correct $key$ message must contain $n-f$ signatures on $echo$ messages from the same view with the same value.
Keys and locks are considered automatically correct if they are from before the first view.
In addition, when checking if a key is correct, parties also check that the key's value is externally valid.

Below we provide an overview of each of the rounds.
The parties proceed in $5$ rounds.  
The general idea is that parties will first confirm that they all agree on the output of the $\PE$ protocol, set a lock to the output and confirm that they are all locked, commit to the lock and terminate.  
If at any point they see that the $\PE$ failed, then they move onto a new view and announce that they are doing so (with proof).  
\\

\noindent \textbf{Round 1:}  The first round in each view begins with a $\viewChange$ protocol.  
The $\viewChange$ protocol determines which keys parties input into the $\PE$ protocol.
To begin, send the current key to all other parties in a $suggest$ message.  
Upon receiving $n-f$ keys, choose the key from the most recent view and input it to the $\PE$ protocol.
\\

\noindent \textbf{Round 2:} The second round proceeds differently depending on which messages parties receive.  This is the round where parties determine whether the $\PE$ was successful or not.
    \begin{itemize}
    \item Upon receiving a value output from another party from the $\PE$ protocol, if that value is correct then echo that message to all other parties.   
    \item If that value is incorrect then send a $blame$ message and proof to all other parties, including a proof that the value was the output from the $\PE$ protocol and that it is incorrect and proceed to the next view.
    The $\PE$ protocol uses an external-validity function that guarantees that all outputs are well-formed and provide correct proofs of their keys.
    However, checking whether the message should be accepted using the local $lock$ fields cannot be modeled as an external validity function, since it is dependent on the running party's local state.
    Therefore, $blame$ messages inform other parties that the $\PE$ protocol output a key which was insufficient to open the local $lock$, and include the local $lock$ fields with proofs that they have been correctly set. 
    If the $\PE$ protocol was successful then the output values should always be correct and open any lock.

    \item Upon receiving a correct $blame$ message and proof, send the $blame$ message to all parties and proceed to the next view.
    
    \item Upon receiving $echo$ messages with two different correct values and proofs that they were outputs of the $\PE$ protocol, send an $equivocate$ message and proof to all parties, and proceed to the next view.  
    If the $\PE$ protocol was unsuccessful then there could be two parties with different correct values, and thus the next view will be necessary to reach agreement.
    
    \item Upon receiving an $equivocate$ message with different values and correct proofs, forward that message, and proceed to the next view.
\end{itemize}

\noindent \textbf{Round 3:} In this round parties are confirming that they believe that the $\PE$ protocol terminated successfully.  Upon receiving $n-f$ $echo$ messages, update the $key$ field before sending a $key$ message to all parties.  
\\

\noindent \textbf{Round 4:} Upon receiving $n-f$ $key$ messages, update the $lock$ field before sending a $lock$ message to all parties.
Setting a $lock$ is the main way the protocol guarantees safety.
As will be stated in the next round, before committing to a value, every party waits to see that at least $n-f$ parties set their locks.
This guarantees that at least $f+1$ nonfaulty parties will have set their locks.
These parties will act as sentinels and won't let any other value get past the $echo$ phase in any future view.
This in turn will make sure that no correct key is set in later views that might allow one of those sentinels to open their lock.
Crucially, before setting a lock, every party makes sure that at least $f+1$ nonfaulty parties set their keys to the current value.
By doing that, every party guarantees that when choosing which value and key to input to the $\PE$ protocol, all nonfaulty parties will hear of the current value and will be capable of opening any older $lock$ a nonfaulty party might have.
\\

\noindent \textbf{Round 5:} 
    If a single honest party begins the final round then the protocol will eventually terminate.  There are two means of termination:  either you see that enough parties are locked, or you see that one other party is (correctly) committed.
    Upon receiving $n-f$ $lock$ messages, send a $commit$ message to all parties and terminate. Upon receiving a $commit$ message with proof that it was sent after receiving enough $lock$ messages, forward that message to all other parties and terminate.
\\

\begin{algorithm}\caption{$\AsyncHS(x_i)$}\label{alg:asynchs}
\begin{algorithmic}[1]
    \State $key_i\gets 0, key\_val_i\gets \perp, key\_proof\gets\perp$
    \State $lock_i\gets 0, lock\_val_i\gets\perp, lock\_proof_i\gets\perp$
    \State $view_i\gets 1$
    \State continually run $\mathsf{checkTermination}()$
    \While{true}
        \State $cur\_view\gets view_i$
        \AsLongAs{$cur\_view=view_i$}\label{line:asLongAs}
            \State delay any message from any view $v$ such that $v>view_i$
            \State call $\viewChange(view_i)$ \Comment{perform first lines in $\viewChange$ before continuing to next line}
            \State continually run $\processMessages(view_i)$ and $\processErrors(view_i)$
        \EndAsLongAs
    \EndWhile
\end{algorithmic}
\end{algorithm}
In the $\AsyncHS$ protocol, it is important to note that we explicitly run the $\mathsf{checkTermination}$ protocol before line~\ref{line:asLongAs}, but the $\processMessages$ and $\processErrors$ protocols after it.
This means that the $\mathsf{checkTermination}$ protocol always runs in the background, whereas once $cur\_view\neq view_i$ party $i$ stops processing messages from $cur\_view$ in $\processMessages$ and $\processErrors$ (and thus don't update their $key$ or $lock$ fields according to messages received in older views).

\begin{algorithm}\caption{$\mathsf{checkTermination}()$}\label{alg:termination}
\begin{algorithmic}[1]
    \Upon{receiving the first $\langle commit,v,\pi_{commit},view\rangle$ message from $j$}
        \If{$\commitCorrect(view,v,\pi_{commit})=1$}
            \State send $\langle commit,v,\pi_{commit},view\rangle$ to every party $j\in[n]$\label{line:commitEcho}
            \State \textbf{output} $v$ and \textbf{terminate}
        \EndIf
    \EndUpon
\end{algorithmic}
\end{algorithm}
\begin{algorithm}\caption{$\viewChange(view)$}\label{alg:viewChange}
\begin{algorithmic}[1]
    \State $suggestions\gets\emptyset$ \Comment{$suggestions$ is a multiset}
    \State send $\langle suggest,key_i,key\_val_i,key\_proof_i,view\rangle$ to every party $j\in[n]$
    \Upon{receiving the first $\langle suggest, k, v, \pi_{key}, view\rangle$ message from party $j$}
        \If{$\keyCorrect(k,v,\pi_{key})=1$ and $k<view$}
            \State $suggestions\gets suggestions\cup\{(k,v,\pi_{key})\}$
            \If{$\left|suggestions\right|=n-f$}
                \State $(k,v,\pi_{key})\gets argmax_{(k,v,\pi_{key})\in suggestions}\{k\}$\Comment{break ties arbitrarily}
                \If{$k=0$}
                    \State $(k,v,\pi_{key})\gets(0,x_i,\perp)$
                \EndIf
                \State \textbf{call} $\PE_{i,view}((k,v,\pi_{key}))$ with the external validity function $\keyCorrect$
            \EndIf
        \EndIf 
    \EndUpon
\end{algorithmic}
\end{algorithm}

\begin{algorithm}\caption{$\processErrors(view)$} \label{alg:faults}
\begin{algorithmic}
    \Upon{receiving the first $\langle blame, k, v, \pi_{key}, \pi_{election}, l, lv, \pi_{lock}, view\rangle$ message from $j$}
        \If{$\lockCorrect(l,lv,\pi_{lock})=1$ and $view\leq k\lor k<l$}
            \Upon{$\PEVerify_{i,view}((k,v,\pi_{key}),\pi_{election})$ terminating}
                \State send $\langle blame, k, v, \pi_{key}, \pi_{election}, l, lv,  \pi_{lock}, view\rangle$ to every party $j\in[n]$ \label{line:blameEcho}
                \State $view_i\gets view_i+1$
            \EndUpon
        \EndIf
    \EndUpon
    \Upon{receiving the first $\langle equivocate, k,v,\pi_{key},\pi_{election},k',v',\pi'_{key},\pi'_{election},view\rangle$ message from $j$}
        \If{$(k,v,\pi_{key})\neq (k',v',\pi'_{key})$}
            \Upon{$\PEVerify_{i,view}((k,v,\pi_{key}),\pi_{election})$ and $\PEVerify_{i,view}((k',v',\pi'_{key}),\pi'_{election})$ terminating}
                \State send $\langle equivocate, k,v,\pi_{key},\pi_{election},k',v',\pi'_{key},\pi'_{election},view\rangle$ to every party $j\in[n]$ \label{line:equivocateEcho}
                \State $view_i\gets view_i+1$
            \EndUpon
        \EndIf
    \EndUpon
\end{algorithmic}
\end{algorithm}

\begin{algorithm}\caption{$\processMessages(view)$}\label{alg:processMessages}
\begin{algorithmic}[1]
    \State $echoes\gets\emptyset$, $keys\gets\emptyset$, $locks\gets\emptyset$
    
    \Upon{$\PE_{i,view}$ outputting $(k,v,\pi_{key}),\pi_{election}$} \Comment{continue updating state according to $\PE_{i,view}$}
        \If{$view>k\geq lock_i$}
            \State $\sigma\gets \sign (sk_i,\langle echo, v, view\rangle)$
            \State send $\langle echo, k, v, \pi_{key}, \pi_{election}, \sigma, view\rangle$ to every party $j\in[n]$
        \Else
            \State send $\langle blame, k, v, \pi_{key}, \pi_{election}, lock_i, lock\_val_i, lock\_proof_i, view\rangle$ to every party $j\in[n]$ \label{line:sendBlame}
            \State $view_i\gets view_i+1$
        \EndIf
    \EndUpon
    \Upon{receiving the first $\langle echo, k, v, \pi_{key},\pi_{election},\sigma, view\rangle$ message from $j$}
        \If{$\verifySignature(pk_j,\langle echo,v,view\rangle,\sigma)=1$}
            \Upon{$\PEVerify_{i,view}((k,v,\pi_{key}),\pi_{election})$ terminating}
                \If{$\exists (k',v',\pi'_{key},\pi'_{election},\sigma',j')\in  echoes \ s.t.\ (k,v,\pi_{key})\neq(k',v',\pi'_{key})$}\label{line:equivocation}
                    \State send $\langle equivocate, k,v,\pi_{key},\pi_{election},k',v',\pi'_{key},\pi'_{election},view\rangle$ to every party $j\in[n]$\label{line:sendEquivocate}
                    \State $view_i\gets view_i+1$
                \Else
                    \State $echoes\gets echoes\cup(k,v,\pi_{key},\pi_{election},\sigma,j)$
                    \If{$\left|echoes\right|=n-f$}
                        \State $sigs \gets \{(\sigma,j)|(k,v,\pi_{key},\pi_{election},\sigma,j)\in echoes\}$
                        \State $key_i\gets view,key\_proof_i\gets sigs,key\_val_i\gets v$
                        \State $\sigma\gets \sign(sk_i,\langle key, v, view\rangle)$
                        \State send $\langle key, v, sigs, \sigma, view \rangle$ to every party $j\in[n]$
                    \EndIf
                \EndIf
            \EndUpon
        \EndIf
    \EndUpon
    \Upon{receiving the first $\langle key,v,\pi_{key},\sigma,view\rangle$ message from $j$}
        \If{$\verifySignature(pk_j,\langle key, v, view\rangle,\sigma)=1$ and $\keyCorrect(view,v,\pi_{key})=1$}
            \State $keys\gets keys\cup\{(\sigma,j)\}$
            \If{$\left|keys\right|=n-f$}
                \State $lock_i\gets view, lock\_proof_i\gets keys, lock\_val_i\gets v$
                \State $\sigma\gets \sign(sk_i,\langle lock,v,view\rangle)$
                \State send $\langle lock, v, lock\_proof_i, \sigma, view\rangle$ to every party $j\in[n]$
            \EndIf
        \EndIf
    \EndUpon
    \Upon{receiving the first $\langle lock,v,\pi_{lock},\sigma,view\rangle$ message from $j$}
        \If{$\verifySignature(pk_j,\langle lock, v, view\rangle,\sigma)=1$ and $\lockCorrect(view,v,\pi_{lock})=1$}
            \State $locks\gets locks\cup\{(\sigma,j)\}$
            \If{$\left|locks\right|=n-f$}
                \State send $\langle commit, v, locks, view\rangle$ to every party $j\in[n]$ \label{line:commit}
                \State \textbf{output} $v$ and \textbf{terminate}
            \EndIf
        \EndIf
    \EndUpon
\end{algorithmic}
\end{algorithm}

\begin{algorithm}\caption{$\keyCorrect(view,v,\pi_{key})$}\label{alg:keyCorrect}
\begin{algorithmic}
    \If{$\validate(v)=0$}
        \State \textbf{return} 0
    \EndIf
    \If{$view=0$}
        \State \textbf{return} 1
    \EndIf
    \If{$\left|\{j|\exists (\sigma,j)\in\pi_{key}\}\right|\geq n-f$ and $\forall (\sigma,j)\in \pi_{key}\ \verifySignature(pk_j,\langle echo, v, view\rangle, \sigma)=1$}
        \State \textbf{return} 1
    \Else
        \State \textbf{return} 0
    \EndIf
\end{algorithmic}
\end{algorithm}

\begin{algorithm}\caption{$\lockCorrect(view,v,\pi_{lock})$}\label{alg:lockCorrect}
\begin{algorithmic}
    \If{$view=0$}
        \State \textbf{return} 1
    \EndIf
    \If{$\left|\{j|\exists (\sigma,j)\in\pi_{lock}\}\right|\geq n-f$ and $\forall (\sigma,j)\in \pi_{lock}\ \verifySignature(pk_j,\langle key, v, view\rangle, \sigma)=1$}
        \State \textbf{return} 1
    \Else
        \State \textbf{return} 0
    \EndIf
\end{algorithmic}
\end{algorithm}

\begin{algorithm}\caption{$\commitCorrect(view,v,\pi_{commit})$}\label{alg:commitCorrect}
\begin{algorithmic}
    \If{$\left|\{j|\exists (\sigma,j)\in\pi_{commit}\}\right|\geq n-f$ and $\forall (\sigma,j)\in \pi_{commit}\ \verifySignature(pk_j,\langle lock, v, view\rangle, \sigma)=1$}
        \State \textbf{return} 1
    \Else
        \State \textbf{return} 0
    \EndIf
\end{algorithmic}
\end{algorithm}

\subsection{Security Analysis}
Our main theorem for demonstrating the security of $\AsyncHS$ is given in \cref{thm:vaba} where we show correctness, validity, termination and quality.
The proof of this theorem relies on several lemmas.

Correctness depends on \cref{lem:safety} where we show that whenever there exists a correct commitment, nonfaulty parties will not send $echo$ messages with values that are inconsistent with this commitment in future views.
The proof of correctness also uses \cref{lem:singleValuePerView} which argues that all nonfaulty parties only send correct messages, and that all correct messages in a given view contain the same value.

Termination depends on Lemmas~\ref{lem:viewProgress} and \ref{lem:goodLeaderTermination}.
\cref{lem:viewProgress} proves that provided no commitment is reached in prior views, honest parties will eventually progress onto the next view.
The proof depends on \cref{lem:correctKeyLock} which argues that nonfaulty parties' local $key$ and $lock$ fields are always correct, and thus will be accepted when received in any message.
\cref{lem:goodLeaderTermination} proves that whenever all non-faulty parties begin a view with valid inputs, the protocol has a constant probability of terminating.
The proof depends on \ref{lem:goodLeaderNoBlame} which argues that nonfaulty parties will not get successfully blamed for their honest inputs.
The proof also depends on the correctness lemmas and \cref{lem:correctKeyLock}.
Validity follows from Correctness and the external validity of the $\PE$.  Quality follows from Termination and the $\alpha$-Binding property of the $\PE$.

We start by defining what it means for a key, lock, or commit to be correct.
\begin{definition}
    A $key$ message of the form $\langle key, v, \pi, \sigma, view\rangle$ is said to be correct if $\keyCorrect(view, v, \pi)=1$.
    Similarly, a $lock$ message of the form $\langle lock, v, \pi, \sigma, view\rangle$ is said to be correct if $\lockCorrect(view, v, \pi)=1$.
    Finally, a $commit$ message of the form $\langle commit, v, \pi, view\rangle$ is said to be correct if $\commitCorrect(view, v, \pi)=1$.
    In addition, the value of each such message is said to be the field $v$.
\end{definition}

The following two lemmas help prove that the protocol maintains safety conditions.
By that we mean that if some nonfaulty party commits to a value, then there will be $f+1$ parties that will act as sentinels in all future views and won't let any other value receive enough $echo$ messages to proceed to late stages of the protocol.
\begin{lemma}\label{lem:singleValuePerView}
    If two messages from a given $view$ are correct, they both have the same value $v$.
    In addition, if a nonfaulty party sends a $key$, a $lock$ or a $commit$ message, then that message is correct.
\end{lemma}

\begin{proof}
    First, observe two correct key messages $\langle key, v, \pi, \sigma, view\rangle$ and $\langle key, v', \pi', \sigma', view\rangle$.
    Since the messages are correct, $\keyCorrect(view,v,\pi)=1$, which means that $\pi$ contains $n-f$ pairs of the form $(\sigma,j)$ with different values $j\in[n]$ such that $\verifySignature(pk_j,\langle echo, v, view\rangle,\sigma)=1$.
    In other words, $\pi$ contains signatures from $n-f$ parties on the message $\langle echo, v, view\rangle$.
    Similarly, $\pi'$ contains signatures from $n-f$ parties on the message $\langle echo, v', view\rangle$.
    Every nonfaulty party sends only one such signature in each view to all parties in an $echo$ message.
    Now, since $2(n-f)=n+(n-2f)\geq n+f+1$, there are at least $f+1$ parties whose signatures are contained in both $\pi$ and $\pi'$, and out of those parties at least one is nonfaulty.
    That nonfaulty party sends only one such message, so $v=v'$.
    Now, before sending a $key$ message, a nonfaulty party $i$ finds that $\left|echoes\right|\geq n-f$.
    Party $i$ only adds a tuple $(k,v,\pi_{key},\pi_{election},\sigma,j)$ to $echoes$ after receiving the first $\langle echo, k, v, \pi_{key}, \pi_{election}, \sigma, view\rangle$ message from $j$ such that $\verifySignature(pk_j,\langle echo, v, view\rangle, \sigma)=1$ and $\PEVerify((k,v,\pi_{key}),\pi_{election})$ terminates.
    Since $\PEVerify$ terminated, $\keyCorrect(k,v,\pi_{key})=1$, and thus $\validate(v)=1$. 
    Otherwise the first condition in $\keyCorrect$ would be true and the output would be $0$ instead.
    If at any point $i$ sees that two such tuples would be added with different values $v\neq v'$, $i$ sends an $equivocate$ message instead and doesn't send a $key$ message.
    Therefore, when sending a message $\langle key, v, \pi, view\rangle$ it does so with $\pi$ containing $n-f$ pairs of the form $(\sigma, j)$ with different values $j$ such that $\verifySignature(pk_j,\langle echo, v, view\rangle, \sigma)=1$ and $\validate(v)=1$, and thus the message is correct.
    
    Now observe two messages $\langle lock, v, \pi, \sigma, view\rangle$ and $\langle lock, v', \pi', \sigma', view\rangle$ such that $\lockCorrect(view,v,\pi)=1$ and $\lockCorrect(view, v',\pi')=1$.
    Similarly to the case above, $\pi$ contains signatures from at least $n-f$ parties on the message $\langle key, v, view\rangle$.
    Out of those $n-f$ parties, at least $f+1$ are nonfaulty.
    Every nonfaulty party $i$ sends only one such signature per view in a $key$ message, and as stated above each $key$ message sent by a nonfaulty party is correct.
    Since the $key$ message is correct, its value is the same as the value of all correct $key$ messages sent in $view$.
    Therefore, comparing the two values $v$ and $v'$ to the value of all correct $key$ messages $v''$, it must be the case that $v=v''=v'$.
    In addition, before sending a message $\langle lock, v, \pi, \sigma, view\rangle$, a nonfaulty party $i$ finds that $\left|keys\right|\geq n-f$.
    Party $i$ only add a pair $(\sigma,j)$ to $keys$ after receiving the first correct $\langle key, v, \pi, \sigma, view\rangle$ message from party $j$ such that $\verifySignature(pk_j,\langle key, v, view\rangle, \sigma)=1$.
    As shown above, all correct $key$ messages in a given $view$ have the same value $v$, so at that point in time $keys$ contains $n-f$ tuples with signatures on the message $\langle key, v, view\rangle$, and thus $i$'s $lock$ message is correct as well.
    The exact same arguments can be made for showing that $commit$ messages have the same value $v$, and that if a nonfaulty party sends a $commit$ message in line~\ref{line:commit} then the message is correct.
    Finally, if a nonfaulty party sends the message $\langle commit, v, \pi, view\rangle$ message in line~\ref{line:commitEcho}, then it first verified that $\commitCorrect(view,v,\pi)=1$, and thus the message is correct as well.
\end{proof}

\begin{lemma}\label{lem:safety}
    If some party sends a $\langle commit, v, \pi, view\rangle$ message such that $\commitCorrect(view,v,\pi)=1$, then for any $view'\geq view$ there exist $f+1$ nonfaulty parties that never send an $\langle echo, k', v', \pi'_{key}, \pi'_{election} \sigma', view'\rangle$ message with $v'\neq v$ .
\end{lemma}
\begin{proof}
    We will prove inductively that for any $view'\geq view$, there must exist $f+1$ such nonfaulty parties.
    First observe $view'=view$.
    Since some party sends a $\langle commit, v, \pi, view\rangle$ message such that $\lockCorrect(view,v,\pi)=1$,
    $\pi$ contains $n-f$ tuples $(\sigma,j)$ with different values $j\in[n]$ such that $\verifySignature(pk_j,\langle lock, v, view\rangle,\sigma)=1$.
    Out of those parties at least one was nonfaulty.
    A nonfaulty party $j$ only sends such a signature $\sigma$ in a $lock$ message.
    Before sending a $lock$ message, $j$ receives $n-f$ correct $key$ messages, and at least one of those was sent by a nonfaulty party $l$.
    From Lemma~\ref{lem:singleValuePerView}, all of those messages contained the same value $v$.
    Before sending that $key$ message, $l$ found that $\left|echoes\right|\geq n-f$.
    Party $l$ only adds a tuple to $echoes$ after receiving the first $echo$ message from each party.
    Before adding a tuple $(k,v,\pi_{key},\pi_{election},\sigma,j)$ to $echoes$, $l$ verifies that there does not exist a tuple $(k',v',\pi'_{key},\pi'_{election},\sigma',j')$ in $echoes$ with $v\neq v'$.
    If such a tuple exists, $l$ finds that the condition in line~\ref{line:equivocation} is true and it sends an $equivocate$ message instead.
    Since it didn't do so, all $n-f$ echo messages it received had the same value $v$ that $l$ sent in its $key$ message.
    Out of those $n-f$ messages, at least $f+1$ were sent by nonfaulty parties.
    Every nonfaulty party sends no more than one $echo$ message to all parties in each view, and thus those $f+1$ parties never send an $echo$ message with any value $v'\neq v$ in $view$.
    
    Assume the claim holds for every $view''$ such that $view'>view''\geq view$.
    Since $\lockCorrect(view,v,\pi)=1$,
    $\pi$ contains $n-f$ tuples $(\sigma,j)$ with different values $j\in[n]$ such that $\verifySignature(pk_j,\langle lock, v, view\rangle,\sigma)=1$.
    Out of those $n-f$ parties, at least $f+1$ are nonfaulty.
    Every nonfaulty party $j$ only sends such a signature $\sigma$ in a $lock$ message.
    In addition, before sending a $lock$ message, every one of those parties sets its $lock_j$ field to $view$.
    Let the set of those nonfaulty parties be $I$.
    It is important to note that the field $lock_j$ only grows throughout the protocol, so every one of the parties $j\in I$ has $lock_j\geq view$ from that point on.
    Now assume by way of contradiction that some party $j\in I$ sent an $\langle echo, k', v', \pi'_{key}, \pi'_{election} \sigma', view'\rangle$ message with $v'\neq v$.
    Before doing that, it output $(k',v',\pi'_{key}),\pi'_{election}$ in $\PE_{i,view}$ such that  $view>k'\geq lock_j\geq view$.
    From the Completeness and External Validity properties of the $\PE$ protocol,   $\keyCorrect(k',v',\pi'_{key})=1$, so $\pi'_{key}$ contains $n-f$ pairs $(\sigma, l)$ such that $\verifySignature(pk_l,\langle echo, v', k'\rangle,\sigma)=1$.
    As discussed above, each nonfaulty party only sends such a signature in an echo message in view $k'$.
    However, $view'>k'\geq view$, so by assumption there exist $f+1$ parties that never send such a message in view $k'$.
    Any set of $n-f$ parties that sent the relevant signatures must have at least one party in common with the $f+1$ parties that never send such a signature, reaching a contradiction.
\end{proof}

The following lemmas show that the system retains liveness and makes progress.
This is done in two parts.
First of all, the first two lemmas show that if some party doesn't terminate in a given view, it eventually reaches the next view.
The next two lemmas then show that if in any view the binding value of the $\PE$ protocol is set to be the input of a party that was nonfaulty when calling the protocol, then if all parties reach that view they terminate in it as well.
The aforementioned event takes place with constant probability, so these two ideas can be combined to show that some party eventually terminates with high probability.
This is done by showing that until this happens, parties advance through different views, and in each one they have a constant probability of terminating.
It is then left to show that once the first nonfaulty party completes the protocol, eventually all nonfaulty parties do as well.
\begin{definition}
    A nonfaulty party $i$ is said to reach a $view$ if at any point its local $view_i$ field equals $view$.
    Similarly, a nonfaulty party $i$ is said to be in $view$ if its local $view_i$ field equals $view$ at that time.
\end{definition}

\begin{lemma}\label{lem:correctKeyLock}
    Let $x_i$ be the input of a nonfaulty party $i$.
    If $\validate(x_i)=1$, then at any point in the protocol $\keyCorrect(key_i,\\key\_val_i,key\_proof_i)=1$.
    In addition, $\lockCorrect(lock_i,lock\_val_i,lock\_proof_i)=1$ at all times in the protocol.
\end{lemma}
\begin{proof}
    If $i$ hasn't updated its local $key_i,key\_val_i,key\_proof$ fields, then $key_i=0$, $key\_val_i=x_i$ and $key\_proof_i=\perp$.
    By assumption $\validate(x_i)=1$, so $\keyCorrect(key_i,key\_val_i,key\_proof\_i)$ doesn't return $0$ when checking whether the value is externally valid and returns $1$ when checking if $key=0$.
    If $i$ updated its local $key_i,key\_val_i,key\_proof_i$ fields in some $view'$, then after doing so it sent the message $\langle key, v, \pi_{key}, \sigma, view'\rangle$, where $v=key\_val_i$, $\pi_{key}=key\_proof_i$ and $view'=key_i$.
    From Lemma~\ref{lem:singleValuePerView}, the message is correct which means that $\keyCorrect(key_i,key\_val_i,key\_proof_i)=1$.
    Similarly, if $i$ hasn't updated its $lock_i$, $lock\_val_i$ and $lock\_proof_i$ fields, then $lock_i=0$ and thus $\lockCorrect(lock_i,lock\_val_i,\\lock\_proof_i)=1$.
    On the other hand, if $i$ updated these local fields, then it sent the message $\langle lock, v, \pi_{lock},\sigma,view\rangle$ afterwards with $v=lock\_val_i$, $\pi_{lock}=lock\_proof_i$ and $view'=lock_i$.
    From Lemma~\ref{lem:singleValuePerView}, the message is correct and thus $\lockCorrect(lock_i,lock\_val_i,lock\_proof_i)=1$.
\end{proof}

\begin{lemma}\label{lem:viewProgress}
    If every nonfaulty party $i$ has an input $x_i$ such that $\validate(x_i)=1$, all nonfaulty parties participate in the protocol, and no nonfaulty party terminates during any $view'$ such that  $view'<view$, then all nonfaulty parties reach $view$.
\end{lemma}
\begin{proof}
    We will prove the claim inductively on $view$.
    First, all nonfaulty parties start in $view=1$.
    Now observe some $view>1$ and assume no nonfaulty party sends a $\langle commit, v, \pi, view'\rangle$ message in line~\ref{line:commit} for any $view'< view$.
    If some nonfaulty party did send such a message in line~\ref{line:commit}, then it did so in $view'$, and terminated immediately afterwards, contradicting the conditions of the lemma.
    By the induction hypothesis, all nonfaulty parties reach $view-1$.
    If some nonfaulty party $i$ sends the message $\langle blame,k,v,\pi_{key},\pi_{election},lock_i,lock\_val_i,lock\_proof_i,view-1\rangle$ in line~\ref{line:sendBlame}, it increments $view_i$ from $view-1$ to $view$.
    Party $i$ only sends such a message if it outputs $(k,v,\pi_{key}),\pi_{election}$ in $\PE_{i,view}$ and finds that $view-1\leq k\lor k<lock_i$.
    Every nonfaulty party $j$ that receives that message sees that the same condition holds in the $\processErrors$ algorithm.
    From Lemma~\ref{lem:correctKeyLock}, $j$ also sees that $\lockCorrect(lock_i,lock\_val_i,lock\_proof_i)=1$.
    Finally, from the Completeness property of $\PE$, eventually  $\PEVerify_{j,view}((k,v,\pi_{key}),\pi_{election})$ terminates.
    At that point $j$ forwards the message to all parties and advances $view_j$ from $view-1$ to $view$.
    In addition, if $i$ sends a $\langle blame,k,v,\pi_{key},\pi_{election},l,lv,\pi_{lock},view-1\rangle$ message in line~\ref{line:blameEcho}, it first received the same message and found that $\lockCorrect(l,lv,\pi_{lock})=1$, and that $view-1\leq k\lor k<lock_i$.
    Furthermore, at some point, $\PEVerify_{i,view}((k,v,\pi_{key}),\\\pi_{election})$ terminates.
    After sending the message, $i$ increments $view_i$.
    Every nonfaulty $j$ that receives the message sees that the same conditions hold.
    From the Agreement on Verification property of $\PE$ $j$ eventually also sees that $\PEVerify_{j,view}((k,v,\pi_{key}),\pi_{election})$ terminates, and increments $view_j$.
    
    On the other hand, if at any point $i$ sends an $equivocate$ message with two sets of values $k,v,\pi_{key},\pi_{election}$ and $k',v',\pi'_{key},\pi'_{election}$ in line~\ref{line:sendEquivocate}, then it first received two $echo$ messages $\langle echo,k,v,\pi_{key},\pi_{election},view-1\rangle$ and $\langle echo,k',v',\pi'_{key},\pi'_{election},view-1\rangle$ such that $(k,v,\pi_{key})\neq(k',v',\pi'_{key})$.
    That is because $i$ only sends such a message after trying to add a tuple $(k,v,\pi_{key},\pi_{election},\sigma,j)$ to $echoes$ and finding that there exist some tuple $(k',v',\pi'_{key},\pi'_{election},\sigma',j')$ with $(k,v,\pi_{key})\neq(k',v',\pi'_{key})$.
    Party $i$ only reaches that point in the algorithm after finding that $\PEVerify_{i,view}((k,v,\pi_{key}),\pi_{election})$ and $\PEVerify_{i,view}((k',v',\pi'_{key}),\pi_{election})$ terminated.
    Every nonfaulty party $j$ that receives the message also sees that $(k,v,\pi_{key})\neq(k',v',\pi'_{key})$ in the $\processErrors$ algorithm. 
    From the Agreement on Verification property of $\PE$, eventually $\PEVerify_{j,view}((k,v,\pi_{key}),\pi_{election})$ and $\PEVerify_{j,view}((k',v',\pi'_{key}),\pi_{election})$ terminate as well.
    At that point, $j$ forwards the message and advances $view_i$ from $view-1$ to $view$.
    In addition, if some party $i$ sends an $equivocate$ message in line~\ref{line:equivocateEcho}, it first receives the same message with the values $k,v,\pi_{key},\pi_{election}$ and $k',v',\pi'_{key},\pi'_{election}$ such that $(k,v,\pi_{key})\neq(k',v',\pi'_{key})$ and at some point $\PEVerify_{i,view}((k,v,\pi_{key}),\pi_{election})$ and $\PEVerify_{i,view}((k',v',\pi'_{key}),\pi'_{election})$ terminate.
    After sending the message, $i$ increments $view_i$.
    Every nonfaulty $j$ that receives the message sees that the same conditions hold, and from the Agreement on Verification property eventually sees that $\PEVerify_{j,view}((k,v,\pi_{key}),\pi_{election})$ and $\PEVerify_{j,view}((k',v',\pi'_{key}),\pi'_{election})$ terminate, and increments $view_j$ as well.
    
    Now it is left to show that there exists some nonfaulty party that sends either a $blame$ message or an $equivocate$ message.
    Assume by way of contradiction no nonfaulty party sends either one of those messages.
    Every nonfaulty party $i$ starts $view-1$ by calling $\viewChange(view-1)$ and sending $\langle suggest,k,v,\pi_{key}, view-1\rangle$ to all parties with $k=key_i$, $v=key\_val_i$ and $\pi_{key}=key\_proof_i$.
    Every nonfaulty party receives that message, and from Lemma~\ref{lem:correctKeyLock}, $\keyCorrect(k,v,\pi_{key})=1$ for every one of those messages.
    In addition, no nonfaulty $i$ has $key_i\geq view-1$ at that time because $i$ would only update $key_i$ to some value $view'\geq view-1$ during $view'$.
    After receiving those messages, all nonfaulty parties add an element to $suggestions$ and then find that $\left|suggestions\right|=n-f$, at which point they perform some local computation and participate in $\PE_{i,view-1}$.
    Nonfaulty parties only add a tuple $(k,v,\pi_{key})$ to $suggestions$ if $\keyCorrect(k,v,\pi_{key})=1$, so the same holds for the value they input to $\PE_{i,view-1}$.
    In other words, all nonfaulty parties participate in $\PE$ with externally valid inputs, so from the Termination of Output property of $\PE$, they eventually output some value.
    Observe some nonfaulty $i$ that outputs $(k,v,\pi_{key}),\pi_{election}$ from $\PE_{i,view-1}$.
    Since $i$ doesn't send a $blame$ message, it sends an $\langle echo, k, v, \pi_{key}, \pi_{election}, \sigma, view-1\rangle$ message with $\sigma=\sign (\sk_i,\langle echo,v,view-1\rangle)$.
    This must mean that $view-1>k\geq lock_i$, because otherwise $i$ would have sent a $blame$ message.
    Every nonfaulty party receives that message and sees that $\verifySignature(\pk_i,\langle echo, v, view\rangle, \sigma)=1$ since $\sigma$ is $i$'s signature on that message.
    From the Completeness property of $\PE$, for every nonfaulty $j$ eventually $\PEVerify_{j,view}((k,v,\pi_{key}),\pi_{election})$ terminates, at which point $j$ checks the conditions for sending an $equivocate$ message in line~\ref{line:equivocation}.
    By assumption, party $j$ doesn't send an $equivocate$ message, so it adds an element to $echoes$.
    After adding such an element for every nonfaulty party, $j$ sees that $\left|echoes\right|\geq n-f$ and it sends a $key$ message.
    From Lemma~\ref{lem:singleValuePerView}, every $key$ message sent by a nonfaulty party is correct.
    A nonfaulty party also adds a signature $\sigma$ for the message $\langle key, v, view-1\rangle$ to every $key$ message.
    Therefore every nonfaulty party receives those messages and adds at least $n-f$ elements to $keys$.
    Following similar logic every nonfaulty party then sends a $lock$ message, and every nonfaulty party adds at least $n-f$ elements to $locks$.
    At that point, every nonfaulty party sends a $commit$ message in $view-1$ and terminates.
    However, that is a contradiction to the conditions of the lemma, completing the proof.
\end{proof}

\begin{lemma}\label{lem:goodLeaderNoBlame}
    If a nonfaulty party $i$ inputs $(k,v,\pi_{key})$ to $\PE_{i,view}$, then no party sends a message $\langle blame, k, v, \pi_{key}, \pi_{election},\\ l, lv, \pi_{lock}, view\rangle$ such that  $\lockCorrect(l,lv,\pi_{lock})=1$ and  $view\leq k$ or $k<l$.
\end{lemma}

\begin{proof}
    Assume by way of contradiction some party $j$ sends such a message.
    First of all, note that $i$ only adds a tuple $(k,v,\pi_{key})$ to $suggestions$ if $k<view$.
    Then, when choosing the tuple with the maximal $k$, it chooses one with $k<view$.
    Every nonfaulty party inputs a tuple $(k,v,\pi_{key})$ with $k\geq 0$, and thus if $l=0$, $k\geq l$.
    Otherwise, $j$ sent a message with some $l>0$.
    Now, if $\lockCorrect(l,lv,\pi_{lock})=1$, then $\pi_{lock}$ contains $n-f$ pairs $(\sigma,j)$ with different values $j\in[n]$ such that $\verifySignature(pk_j,\langle key, lv, l\rangle,\sigma)=1$.
    Out of those signatures, at least $f+1$ are from nonfaulty parties.
    Let the set of those nonfaulty parties be $I$.
    Nonfaulty parties only send such a signature in $key$ messages.
    Before sending a $key$ message, each one of the parties $m\in I$ sets its local $key_m$ field to $l$.
    Note that nonfaulty parties only increase their local $key_m$ fields, so from this point on, $key_m\geq l$ for every $m\in I$.
    Now, before $i$ inputs $(k,v,\pi_{key})$ to $\PE_{i,view}$, it sees that $\left|suggestions\right|\geq n-f$.
    Party $i$ only adds elements to $suggestions$ after receiving the first $\langle suggest, k, v, \pi_{key}, view\rangle$ message from each party.
    Therefore, $i$ adds tuples to $suggestions$ as a result of receiving such a message from at least $n-f$ parties.
    There are $f+1$ parties in $I$, and $i$ received $suggest$ messages from $n-f$ different parties, so at least one of the parties from which it received $suggest$ messages is in $I$.
    Let $m\in I$ be that party. 
    Party $m$ sends its local fields $key_m$, $key\_val_m$ and $key\_proof_m$ in its $suggest$ message.
    As shown above, $key_m\geq l$, so when computing which value to input to $\PE_{i,view}$, $i$ has at least one tuple $(k,v,\pi_{key})\in suggestions$ such that $k\geq l$.
    When choosing which value to input, $i$ takes the tuple with the largest value $k$, so its choice $(k,v,\pi_{key})$ must have $k\geq l$, completing the proof.
\end{proof}

\begin{lemma}\label{lem:goodLeaderTermination}
    If all nonfaulty parties start $view$ and every nonfaulty $i$ has input $x_i$ such that $\validate(x_i)=1$, then with constant probability all nonfaulty parties terminate during $view$.
\end{lemma}
\begin{proof}
    If at any point some nonfaulty party terminates, it must have sent a $commit$ message to all parties.
    From Lemma~\ref{lem:singleValuePerView} that message is correct, so all nonfaulty parties receive the message and terminate as well.
    From this point on we will not deal some of the parties terminating early in $view$ and some not terminating at all.
    The first thing that a nonfaulty party does in $view$ is calling $\viewChange$ and sending a $suggest$ message to every party with the local fields $key_i$, $key\_val_i$ and $key\_proof_i$.
    From Lemma~\ref{lem:correctKeyLock}, $\keyCorrect(key_i,key\_val_i,key\_proof_i)=1$.
    Therefore, when a nonfaulty party $j$ receives that message, it adds a tuple to $suggestions$.
    After receiving such a message from every nonfaulty party, $j$ finds that $\left|suggestions\right|\geq n-f$, and it starts participating in $\PE_{i,view}$ after choosing a tuple from $suggestions$ as an input.
    Before a nonfaulty party sends a $blame$ or an $equivocate$ message it must either output a value from $\PE_{i,view}$, or find that $\PEVerify_{i,view}$ terminates for some value.
    Both of those things only happen after completing $\PE_{i,view}$.
    In other words, all nonfaulty parties participate in $\PE$ and wait for it to terminate before any of them proceed to the next view.
    Before adding a tuple $(k,v,\pi_{key})$ to suggestions, every nonfaulty $i$ checks that $\keyCorrect(k,v,\pi_{key})=1$, and since all nonfaulty parties participate in the $\PE$ protocol with inputs they chose from $suggestions$, their input is externally valid.
    Combining those two observations, from the Termination of Output property of $\PE$, all nonfaulty parties eventually output some value when running $\PE$.
    Now the lemma is proven by proving a closely related claim.
    If in $view$ the binding value $x^*$ of $\PE$ as defined in the $\alpha$-Binding property of the $\PE$ protocol is the input of some party that acted in a nonfaulty manner when it started the $\PE$ protocol, then all parties terminate during $view$.
    From the $\alpha$-Binding property of $\PE$ this event happens with probability $\alpha$ ($\alpha=\frac{1}{3}$ in our implementation), so all parties terminate during $view$ with a constant probability.
    
    If the the binding value is indeed the input of a party that acted in a nonfaulty manner when it started $\PE$, then from the Binding Verification property of $\PE$ there is exactly one tuple $(k,v,\pi_{key})$ for which it is possible that $\PEVerify_{i,view}((k,v,\pi_{key}),\pi_{election})$ terminates for a nonfaulty $i$.
    This prevents a nonfaulty party from sending an $equivocate$ message in line~\ref{line:sendEquivocate} because only tuples with those values could be in $echoes$.
    In addition, this prevents a nonfaulty $i$ from sending an $equivocate$ message in line~\ref{line:equivocateEcho} because then if the tuples $(k,v,\pi_{key})$ and $(k',v',\pi'_{key})$ are different, $\PEVerify_{i,view}$ would not terminate for at least one of the tuples.
    If the aforementioned event take place, from the Completeness and Binding Verification properties of $\PE$ every nonfaulty party outputs the tuple $(k,v,\pi_{key})$, with some proof $\pi_{election}$, such that $(k,v,\pi_{key})$ was the input of a nonfaulty party $j$ to $\PE$.
    We would now like to show that no nonfaulty party $i$ sends a $blame$ message in $view$.
    Before sending a $blame$ message in line~\ref{line:sendBlame}, $i$ makes sure that $view\leq k\lor k<lock_i$.
    Also, from Lemma~\ref{lem:correctKeyLock}, $\lockCorrect(lock_i,lock\_val_i,lock\_proof_i)=1$.
    This means that if $i$ sends a $\langle blame, k, v, \pi_{key}, \pi_{election}, l, lv, \pi_{lock},view\rangle$ message in line~\ref{line:sendBlame} it does so with $view\leq k\lor k<lock_i$ and $\lockCorrect(lock_i,lock\_val_i,lock\_proof)=1$.
    Since $(k, v, \pi_{key})$ was some nonfaulty party's input to the $\PE$ protocol, this contradicts Lemma~\ref{lem:goodLeaderNoBlame}.
    Similarly, no nonfaulty party $i$ sends a $blame$ message in line~\ref{line:blameEcho}, because before doing so it checks that the same conditions hold and that $\PEVerify_{i,view}((k,v,\pi_{key}),\pi_{election})$ terminates.
    As stated above, $\PEVerify$ only terminates on the tuple $(k, v, \pi_{key})$ which is some nonfaulty party's input to $\PEVerify$, reaching the same contradiction. 
    
    Nonfaulty parties only proceed to $view+1$ after sending either a $blame$ or an $equivocate$ message, so no nonfaulty party proceeds to $view+1$.
    Since no nonfaulty party sends a $blame$ message, each one sends an $\langle echo, k, v, \pi_{key}, \pi_{election}, \sigma,\\ view\rangle$ message after completing the $\PE_{i,view}$ call, with $\sigma$ being a signature on the message $\langle echo, v, view\rangle$.
    When receiving the message, every nonfaulty party $j$ sees that $\sigma$ is indeed a signature on $\langle echo, v, view\rangle$.
    Then, from the Completeness property of $\PE$, $\PEVerify_{j,view}((k,v,\pi_{key}),\pi_{election})$ eventually terminates.
    Since $j$ doesn't send an $equivocate$ message in $view$, it then adds a tuple to $echoes$.
    After such a tuple is added for every nonfaulty party, $j$ sees that $\left|echoes\right|=n-f$ and it sends a message $\langle key, v, \pi_{key}, \sigma, view\rangle$ to all parties with $\sigma$ being a signature on $\langle key, v, view\rangle$.
    From Lemma~\ref{lem:singleValuePerView}, that message is correct.
    Therefore, when receiving that message, every nonfaulty party sees that the message is correct and that $\sigma$ is a signature on $\langle key, v, view\rangle$, and adds a pair $(\sigma,i)$ to $keys$.
    After adding such a pair for every nonfaulty party, $j$ has $\left| keys\right|=n-f$ and it sends a $lock$ message.
    Using identical arguments, eventually every nonfaulty party sends a $commit$ message and terminates if it hasn't done so earlier.
\end{proof}

\begin{theorem}\label{thm:vaba}
    Protocol $\AsyncHS$ is a Validated Asynchronous Byzantine Agreement protocol resilient to $f<\frac{n}{3}$ Byzantine parties.
\end{theorem}

\begin{proof}
    Each property is proven individually.
    
    \textbf{Correctness.} If some nonfaulty party outputs the value $v$ in $view$, it first sends a $\langle commit, v, \pi, view\rangle$ message.
    Let $view$ be the first view  (i.e. the one with the lowest value) such that some nonfaulty party sends a $\langle commit, v, \pi, view\rangle$ message.
    First of all, from Lemma~\ref{lem:singleValuePerView}, nonfaulty parties only send correct $commit$ messages, so $\langle commit, v, \pi, view\rangle$ is a correct $commit$ message.
    Now observe some message $\langle key, v', \pi', \sigma', view'\rangle$ such that $\keyCorrect(view', v', \pi')=1$ and $view'\geq view$.
    Since $\keyCorrect(view', v', \pi')=1$, $\pi'$ contains $n-f$ pairs $(\sigma,j)$ with different values $j\in[n]$ such that $\verifySignature(pk_j,\langle echo, v', view'\rangle,\sigma)=1$.
    Nonfaulty parties only send such a signature $\sigma$ in an $echo$ message.
    From Lemma~\ref{lem:safety}, in any $view'\geq view$ there exist $f+1$ nonfaulty parties that never send an $echo$ message with any value $v'\neq v$.
    Out of the $n-f$ parties whose signatures are in $\pi'$, at least one is from one of the $f+1$ parties that never sends an $echo$ message with any value $v'\neq v$ in $view'$.
    Therefore, it must be the case that $v'=v$.
    Now, assume some nonfaulty party $i$ sends a $commit$ message in $view'$.
    Before doing so it receives $n-f$ correct $lock$ messages, at least one of which was sent by a nonfaulty party.
    Before sending that $lock$ message, the nonfaulty party receives $n-f$ correct $key$ messages.
    As discussed above, that key message has the value $v$.
    From Lemma~\ref{lem:singleValuePerView}, $i$ sends a correct $commit$ message because it is nonfaulty, and every correct $commit$ message sent in $view'$ has the same value $v$.
    Finally, after sending the $commit$ message, $i$ outputs $v$ and terminates.
    Therefore, all nonfaulty parties that output some value must output the value $v$.
    
    \textbf{Validity.} If some nonfaulty party $i$ outputs a value $v$, it first sends a $\langle commit, v, \pi, view\rangle$ message.
    As discussed in the proof of the Correctness property, at least $n-f$ parties sent $key$ messages in $view$ with the value $v$ as well.
    At least one of those parties is nonfaulty.
    Party $i$ only sends a $\langle key, v, \pi, \sigma \rangle$ message after receiving an $\langle echo, k, v, \pi_{key}, \pi_{election}, \sigma, view\rangle$ message such that $\PEVerify_{i,view}((k,v,\pi_{key}),\pi_{election})$ terminates.
    From the External Validity property of $\PE$, this means that $\keyCorrect(k,v,\pi_{key})=1$.
    Now, if $\validate(v)=0$, $\keyCorrect(k,v,\pi_{key})=0$, so it must be the case that $\validate(v)=1$.
    
    \textbf{Termination.} 
    If at any point a nonfaulty party terminates it sends a $\langle commit, v, \pi, view\rangle$ message.
    From Lemma~\ref{lem:singleValuePerView} the message is correct, so all nonfaulty parties eventually receive the message and terminate as well.
    Now assume that every nonfaulty party $i$ has an input $x_i$ such that $\validate(x_i)=1$ and that all nonfaulty parties participate in the protocol.
    Observe some $view$, and assume no nonfaulty party terminated during $view'$ for any $view'<view$.
    In that case, from Lemma~\ref{lem:viewProgress} all nonfaulty parties eventually reach $view$.
    Then, from Lemma~\ref{lem:goodLeaderTermination}, with constant probability all nonfaulty parties terminate during $view$.
    In order for a nonfaulty party not to terminate by $view$, that constant probability event must not have happened in each one of the previous views.
    The nonfaulty parties run the $\PE$ protocol with independent randomness in each view and thus for any adversary's strategy, there is an independent constant probability of terminating in each view. 
    Therefore, the probability of reaching a given view decreases exponentially with the view number and thus approaches $0$ as $view$ grows.
    In other words, all nonfaulty parties almost-surely terminate.
    
    \textbf{Quality.} Assume some nonfaulty party completed the protocol, otherwise the claim holds trivially.
    This means that it at least completed the $\PE$ protocol in $view=1$.
    From the $\alpha$-Binding property of $\PE$, with probability $\alpha$ or greater the binding value is the input of some party that behaved in a nonfaulty manner when starting $\PE$.
    Let $i$ be that party and $(k,v,\pi)$ be its input to the protocol.
    Using the same arguments as the ones made in Lemma~\ref{lem:goodLeaderTermination}, in that case no nonfaulty party sends a $blame$ or an $equivocate$ message during $view$.
    Then, following similar logic to the one in Lemma~\ref{lem:goodLeaderTermination}, every nonfaulty party that hasn't committed due to a message from an earlier view eventually terminates after sending a $commit$ message with the value $v$ proposed by party $i$.
    No party can commit due to a message from an earlier view because there is no earlier view.
    Therefore, every nonfaulty party that participates in $view$ and outputs a value from $\PE$, terminates and outputs the value $v$ that $i$ proposed.
    Before sending its proposal, $i$ sees that $\left|suggestions\right|=n-f$.
    Party $i$ only adds a tuple to $suggestions$ after receiving the first $\langle suggest, k, v, \pi, view\rangle$ message from each party $j\in[n]$.
    Each of those tuples must have $k<view=1$.
    At that point no nonfaulty party updated its $key_j$, $key\_val_j$ and $key\_proof_j$ fields, so they send messages with $k=0$.
    Since at least one of the $n-f$ messages was sent by a nonfaulty party, there exists some $(k,v,\pi)\in suggestions$ such that $k=0$, and as shown above there is no such tuple with $k>0$.
    Therefore, when computing its input to $\PE_{i,1}$, $i$ sees that the tuple with maximal $k$ in $suggestions$ has $k=0$.
    Party $i$ then uses $(0,x_i,\perp)$ as input to $\PE$, with $x_i$ being its input to the $\AsyncHS$ protocol.
    As shown above, with constant probability all nonfaulty parties that start $view$ output $x_i$, completing the proof.
\end{proof}

\section{Asynchronous Distributed Key Generation}\label{sec:dkg}

The protocol is a simple construction of an Asynchronous Distributed Key Generation protocol using a Validated Asynchronous Byzantine Agreement protocol.
Parties start off by sending each other DKG shares.
After receiving such a share from at $n-f$ parties, every party aggregates the shares, and inputs the aggregated DKG transcript into the $\AsyncHS$ protocol.
The protocol is called with an external validity function checking whether a DKG transcript is valid.
After completing the $\AsyncHS$ protocol with some output $\dkg$, all parties complete the $\ADKG$ protocol, outputting the same value.
From the properties of the $\AsyncHS$ protocol, all parties eventually output the same DKG transcript, and since it must be externally valid, that transcript verifies.

\begin{algorithm}\caption{$\ADKG_i$}
\begin{algorithmic}[1]
    \State $shares\gets\emptyset$ \Comment{$shares$ is a multiset}
    \ForAll{$j\in[n]$}
        \State $\share_{i,j}\gets \generateShare(\sk_i)$
        \State send $\langle \share_{i,j}\rangle$ to party $j$
    \EndFor
    \Upon{receiving the first $\langle share_{j,i}\rangle$ message from $j$}
        \If{$\verifySecret(\pk_j,\share_{j,i})=1$}
            \State $shares\gets shares\cup \{share_{j,i}\}$
            \If{$\left|shares\right|=n-f$}
                \State $\proposal\gets \aggregateShares(shares)$
                \State \textbf{call} $\AsyncHS$ with input $\proposal$ and external validity function $\DKGVerify$
            \EndIf
        \EndIf
    \EndUpon
    \Upon{$\AsyncHS$ terminating with output $\dkg$}
        \State \textbf{output} $\dkg$ and \textbf{terminate}
    \EndUpon
\end{algorithmic}
\end{algorithm}

\begin{theorem}
    Protocol $\ADKG$ is an Asynchronous Distributed Key Generation protocol resilient to $f<\frac{n}{3}$ Byzantine parties.
\end{theorem}
\begin{proof}
Each property is proven individually.

\textbf{Security Preservation.} We see that if $(\generateShare, \verifySecret, \aggregateShares, \DKGVerify)$ satisfies security preservation with regard to a concurrent adversary for some threshold application, then $\ADKG$ also satisfies security preservation for the same application.
Indeed, should our adversary expect to receive an honest DKG share at any point in the protocol, then this can be modelled as an adversary making concurrent requests to a $\generateShare$ oracle.

\textbf{Correctness.}  Follows immediately from the correctness of $(\generateShare, \verifySecret, \aggregateShares, \DKGVerify)$.

\textbf{Agreement.} If two nonfaulty parties $i,j$ complete the protocol with the outputs $\dkg,\dkg'$, then they first completed the $\AsyncHS$ protocol with that same output.
By the Agreement property of the $\AsyncHS$ protocol, $\dkg=\dkg'$.
Furthermore, from the Validity property of the $\AsyncHS$ protocol, $\DKGVerify(\dkg)=1$.

\textbf{Termination.} If all nonfaulty parties participate in the protocol, they all send a share of a DKG to all parties.
Every nonfaulty party $i$ then receives a message $\langle \share_{j,i}\rangle$ from every nonfaulty party $j$, sees that $\verifySecret(\pk_j,\share_{j,i})=1$ and adds it to $shares$.
After adding such a value for every nonfaulty party, $i$ sees that $\left|shares\right|=n-f$, it aggregates the shares to a single proposal, and starts participating in $\AsyncHS$ with that proposal.
Note that $\proposal$ is an aggregation of $n-f$ shares $\share_{j,i}$ such that $\verifySecret(\pk_j,\share_{j,i})=1$, and thus $\DKGVerify(\proposal)=1$.
All nonfaulty parties use $\DKGVerify$ as their external validity function, so every nonfaulty party has an externally valid input.
Therefore, from the Termination property of $\AsyncHS$, all parties almost-surely complete $\AsyncHS$, output some value, and terminate.

\end{proof}
\section{Efficiency of our Protocols Assuming Concrete Cryptography Algorithms}\label{sec:efficiency}
In this section we make suggestions as to which cryptography algorithms to instantiate our Broadcast, Gather, Proposal Election, No Waitin' HotStuff, and A-DKG protocols with.
We then analyse the efficiency of our protocols under the suggested cryptography algorithms.
Unlike in the introduction we will keep track of a cryptographic security parameter $\lambda$ which is the number of bits required to ensure the cryptographic algorithm is secure against computational adversaries.

\subsection{Broadcast}\label{sec:efficiency:broadcast}
All our protocols rely on the use of an asynchronous broadcast protocol.
We can instantiate a broadcast protocol for a message of $m$ words where the total number of words sent in all messages is $O(n^2 \log(n) \lambda + m \cdot n)$.

We suggest the use of the a broadcast protocol by Cachin and Tessaro \cite{CachinT05a} described in \cref{sec:broadcast} which relies on a vector commitment.
For the vector commitment we consider using Merkle-Trees.
Merkle trees have commitment size $c = O(\lambda)$, opening proof size $p = O(\log(n) \lambda)$, and concretely are very fast to prove and verify.
Theoretically it is possible to reduce the opening proof size down to $O(1)$ using SNARKs, but this comes at the cost of a trusted setup and concretely high proving time.
The protocol requires a constant number of rounds (3 overall).
The following theorem is proven in \cref{sec:broadcast:proofRBComplexity}.
\begin{theorem}\label{thm:RBComplexity}
To broadcast a message $M$ of size $m$,
the total number of words sent in all messages is $O(n^2\cdot (c+p)+m\cdot n)$ words, where $c$ is the number of words in a commitment and $p$ is the number of words in a proof.
\end{theorem}

\subsection{Verifiable Gather}\label{sec:efficiency:gather}
The $\Gather$ protocol from \cref{sec:gather} relies solely on the existence of a broadcast protocol.
We instantiate $\Gather$ such that the total number of words sent overall is $O(\lambda n^3 \log n + m n^2)$.

We use the broadcast protocol evaluated in \cref{sec:efficiency:broadcast} which has complexity $b(m) = O(n^2 \log(n) \lambda + m \cdot n)$.
Using the result from \cref{thm:efficiency:gather}:
\[O( n b(m) )  
= O( n^3 \log(n) \lambda + m \cdot n^2 ) .
\]
The implementation in this paper requires $3$ broadcast rounds, and each one of those requires a constant number of rounds.
Therefore, overall the $\Gather$ protocol requires a constant number of rounds.

\begin{theorem}\label{thm:efficiency:gather}
    If protocol $\Gather$ is run with inputs of size $m$
    then $O( n b(m) )$ 
    words are sent overall where $b(m)$ is the complexity of a broadcast for  $m$ words.
\end{theorem}
\begin{proof}
    Overall in the protocol, each party broadcasts its input once and vectors of size $n=O(m)$ twice.
    There are $O(n)$ such broadcasts throughout the protocol, so overall the number of words sent is $O( n b(m) )$.
\end{proof}

\subsection{Proposal Election}\label{sec:efficiency:PE}
The $\PE$ protocol from \cref{sec:weakleader} relies on the existence of a gather protocol and a threshold VRF.
We instantiate $\PE$ such that the total number of words send overall is 
$O( \lambda n^3 \log(n)+ m n^2 ) $.

We use the broadcast protocol evaluated in \cref{sec:efficiency:broadcast} which has complexity $b(m) = O(n^2 \log(n) \lambda + m \cdot n)$.
In addition, we use the gather protocol evaluated in \cref{sec:efficiency:gather} which has complexity $g(m) = O(n^3 \log(n) \lambda + m \cdot n^2)$.
For the threshold VRF we suggest the use of the threshold VUF by Gurkan et al.~\cite{GurkanJMMST21}.
In the random oracle model we can then instantiate a threshold VRF by hashing the function evaluation.
This threshold VRF has $d_s = O(\lambda n)$ sized $\dkg$ shares,
$d = O(\lambda n)$ sized $\dkg$s,
$e_s = O(\lambda)$ sized evaluation shares (with their respective proofs),
and $e = O(\lambda)$ sized evaluations.
Using the result from \cref{thm:efficiency:PE}
\[O( n^3\cdot e_s  + n^2 d_s + g(m + d) + b(n) )  
= O( n^3\cdot \lambda  + n^2 \lambda n +  n^3 \log(n) \lambda + (m + \lambda n) \cdot n^2 + \lambda n^3\log(n)+n^3)
= O( \lambda n^3 \log(n)+ m n^2 ) .
\]
The implementation in this paper requires two rounds of point-to-point messages, as well as a single $\Gather$ round and a single broadcast round.
Both the $\Gather$ and broadcast protocols require a constant number of rounds, so this yields a constant-round $\PE$ protocol.

\begin{theorem}\label{thm:efficiency:PE}
    If protocol $\PE$ is run with inputs of size $m$
    then $O( n^3\cdot e_s  + n^2 d_s + g(m + d) + b(n) ) $ 
    words are sent overall, 
    where $g(m)$ is the complexity of a gather for $m$ words, $b(m)$ is the complexity of a broadcast for $m$ words,
    $d_s$ is the size of the DKG shares, $d$ is the size of the DKGs, and $e_s$ is the size of the VRF evaluation shares (and proofs).
\end{theorem}
\begin{proof}
    Every party starts the protocol by sending DKG shares of size $O(d_s)$ to every other party, totalling in $O(n^2 d_s)$ words overall.
    Afterwards, all parties participate in a $\Gather$ protocol 
    with inputs of size $O(m + d)$
    which requires a total of
    $O( g(m + d) )$ words to be sent.
    Following that, parties broadcast sets containing $O(n)$ indices, each requiring a single word.
    Overall, this requires $O(b(n))$ words to be sent.
    Finally, every party $i$ sends messages with an index, an evaluation share, and a proof to every party.
    This is done whenever $i$ outputs a set $X$ with a tuple $\gatherTuple{k}$ from the $\GatherVerify$ protocol such that $\gatherTuple{k}\in X\notin \startEval_i$.
    Immediately after sending such a message, $i$ updates $\startEval_i$ to contain $X$.
    As shown in \cref{lem:startEval}, there is only one such tuple for every $k\in[n]$ in $\startEval_i$, so $i$ sends no more than $n$ such messages.
    Therefore, this requires a total of $O(n^3)$ messages, each containing 
    $O(e_S)$ words.
    Summing all of those terms gives the result.
\end{proof}

\subsection{No Waitin' HotStuff}\label{sec:efficiency:vaba}
The $\AsyncHS$ protocol from \cref{sec:consensus} relies on the existence of a proposal election protocol and a signature scheme.
We instantiate $\AsyncHS$ such that the expected total number of words sent overall is 
$O( \lambda n^3 \log(n) + mn^2 ) $.
The below theorem shows that the total number of words per view is $O( \lambda n^3 \log(n) + mn^2 ) $, and that the total expected number of views is $O(1)$, resulting in an expected $O( \lambda n^3 \log(n) + mn^2 ) $ word complexity overall.
The theorem also shows that each view consists of a constant number of rounds, resulting in a constant expected number of rounds overall.

We use the $\PE$ protocol evaluated in \cref{sec:efficiency:PE} which has complexity 
$p(m) = O( \lambda n^3 \log(n)+ m n^2 )$.
For the signature scheme we suggest the use of Schnorr signatures which have size $s = O(\lambda)$.
Using the result from: \cref{thm:efficiency:vaba}
\[O(s n^3 + m n^2 + p(m))
= O( n^3\cdot \lambda  + m n^2 + \lambda n^3 \log(n)+ m n^2  ) = O( \lambda n^3 \log(n)+ m n^2 ) .
\]

\begin{theorem}\label{thm:efficiency:vaba}
    If protocol $\AsyncHS$ is run with inputs of size $m$
    using the $\PE$ protocol described in \cref{sec:weakleader},
    then all nonfaulty parties terminate in $O(1)$ expected views, where each view consists of a constant number of rounds.
    In addition, the total number of words sent in each view is
    $O(s n^3 + m n^2 + p(m)) $ 
    where $p(m)$ is the complexity of a proposal election for $O(m)$ words
    and $s$ is the size of the signatures.
\end{theorem}
\begin{proof}
    As shown in the proof of the Termination property of the protocol, there is a constant probability $\alpha$ that all nonfaulty parties terminate in $view$ or before it for any one $view$.
    Note that when following the proof of the Termination property, the proof of  \cref{lem:goodLeaderTermination} can actually be used to show that with constant probability no nonfaulty party will ever reach a late view.
    Those probabilities are independent, and thus the number of required views is described by a geometric random variable.
    From well known properties of such variables, the expected number of views required is $\frac{1}{\alpha}$, which is constant.
    
    In each view all nonfaulty parties send a constant number of all-to-all messages in the $suggest$, $echo$, $key$, $lock$ and $commit$ rounds, totalling in $O(n^2)$ messages overall (and possibly $blame$ and $equivocate$ messages).
    Each message contains $m$ words containing a value to be agreed upon, 
    a constant number of additional words and a constant number of proofs.
    Each proof contains $O(n)$ signatures and indices of parties.
    Note that the proof output in our implementation of the $\PE$ protocol also consists of $O(n)$ indices of parties.
    Overall, when not counting the complexity of the $\PE$ protocol, 
    each view in the $\AsyncHS$ protocol requires $O(n^3+(m+ s n) n^2)=O(s n^3 + m n^2 )$
    words.
    Our result is obtained when we add $p(m)$ the complexity of the $\PE$ protocol.
    
    Each view consists of a round of point-to-point communication for sending $suggest$, $echo$, $key$, $lock$ and $commit$ messages (and possibly $blame$ or $equivocate$ messages).
    In addition, all parties call the $\PE$ protocol once per view.
    In the implementation provided above, the $\PE$ protocol requires a constant number of rounds, resulting in a constant number of rounds per view.
\end{proof}

\subsection{Asynchronous Distributed Key Generation}
The A-DKG protocol from \cref{sec:dkg} relies on the existence of a Validated Asynchronous Byzantine Agreement protocol and 
DKG algorithms $\DKG\Share$, $\DKG\Share\Verify$, $\DKG\Aggregate$, $\DKG\Verify$.
We instantiate A-DKG such that the expected total number of words send overall is 
$O( \lambda n^3 \log(n) ) $.

We use the $\AsyncHS$ protocol evaluated in \cref{sec:efficiency:vaba} which has expected word complexity 
$v(m) = O( \lambda n^3 \log(n) + n^2 \cdot m ) $ and the DKG algorithms 
$\DKG\Share$, $\DKG\Share\Verify$, $\DKG\Aggregate$, $\DKG\Verify$ 
from the synchronous DKG of Gurkan et al. \cite{GurkanJMMST21}.
This DKG has $D_s = O(\lambda n)$ sized  $\dkg$ shares and $D = O(\lambda n)$ sized $\dkg$s.
Using the result from \cref{thm:efficiency:vaba}:
\[O(n^2 D_s +  v(D))
= O( n^3 \cdot \lambda  + \lambda n^3 \log(n) + n^2 \cdot (\lambda n)  ) = O( \lambda n^3 \log(n) ) .
\]
The protocol requires a single round of point-to-point communication for sending DKG shares, and a single call to the $\AsyncHS$ protocol.
Since the $\AsyncHS$ protocol requires a constant expected number of rounds, so does the $\ADKG$ protocol.

\begin{theorem}
    If protocol $\ADKG$ is run using the $\AsyncHS$ protocol described in \cref{sec:consensus},
    then all nonfaulty parties terminate in $O(1)$ expected views.
    In addition, the total number of words sent in each view is
    $O( n^2 D_s +  v(D) ) $
    where $v(m)$ is the complexity of a $\AsyncHS$ protocol for $O( m )$ words, $D_S$ is the size of the DKG shares and $D$ is the size of the DKGs.
\end{theorem}
\begin{proof}
    In the beginning of the protocol, all parties send a DKG share of size $O(D_S)$ to all parties, requiring a total of $O(D_s n^2)$ words.
    The parties then call $\AsyncHS$ with an aggregated DKG of size $O(D)$
    words.
    The $\AsyncHS$ protocol requires an expected $v(D)$ words to be sent overall, which gives us our result.
\end{proof}

\bibliographystyle{ACM-Reference-Format}
\bibliography{references}

\newpage
\appendix
\section{Background: Reliable Broadcast for asynchronous systems} \label{sec:broadcast}
Throughout our agreement protocol we shall use a reliable broadcast method by Cachin and Tessaro \cite{CachinT05a} which applies error correcting codes to Bracha broadcast \cite{Bracha87}.
The broadcast protocol has communication complexity $\mathcal{O}(n^2 \log(n) + |M| n)$ for $n$ the total number of participants and $|M|$ the size of the message.
It can tolerate up to $f<\frac{n}{3}$ Byzantine adversaries and works in the asynchronous setting.

\subsection{Construction}
The protocol is extremely similar to Bracha's famous reliable broadcast protocol \cite{Bracha87}.
In Bracha's protocol, the dealer first sends a message $\langle value,m\rangle$ to all parties.
After receiving the first message from the dealer, every nonfaulty party responds with an $\langle echo,m\rangle$ message.
Then, after receiving $n-f$ $\langle echo,m\rangle$ messages, the parties respond with a $\langle ready,m\rangle$ message.
In addition, if some party receives $f+1$ $\langle ready,m\rangle$ messages and it did not send a ready message yet, it also sends a $\langle ready, m\rangle$ message.
Finally, after receiving $n-f$ $\langle ready,m\rangle$ messages, every party outputs $m$ and terminates.

Unfortunately, when sending a large message $M$, every message sent by the parties contains all of $M$, yielding large communication costs.
Cachin \textit{et al.}'s clever approach to reducing the communication costs was employing error correction codes in the form of Reed-Solomon encoding.
Instead of just sending the message $M=(m_0,\ldots,m_\ell)$, the dealer treats the message as coefficients of a polynomial $p(x)=\sum_{k=0}^\ell m_k\cdot x^k$.
Then for every nonfaulty party $j$ the dealer computes a set $P^j$ of $\lceil\frac{\ell+1}{f+1}\rceil$ values on the polynomial $p(x)$.
Then the dealer commits to to the vector $P=(P^i,\ldots,P^n)$, and sends each party $j$ the commitment $com$, the set $P^j$, and a proof $\pi^j$ that the $j$'th element in the committed vector is $P^j$ .
Then, similarly to Bracha's protocol, after receiving a message and checking that the proof is correct, every party sends an echo message with the same information.
Now, after receiving $n-f$ echo messages with the same commitment and correct proofs, every nonfaulty party $j$ should send a ready message with the same commitment, with a set $P^j$ values and with a proof $\pi^j$.
However, $j$ might not have received the set $P^j$ and the proof $\pi^j$, so in order to be able to compute those values, it interpolates the points in $f+1$ of the sets it received $(k,y_k)$ to a polynomial $p$ of degree $\ell$ or less, checks that the commitment is indeed a commitment to a vector $P=(P^1,\ldots,P^n)$ such that each $P^k$ is a set of $\lceil\frac{\ell+1}{f+1}\rceil$ points on the polynomial $p(x)$, and then computes the set of points $P^j$ that it should have received, as well as a proof $\pi^j$ that the $j$'th element in the committed vector is $P^j$ for each one of its points.
After doing that, $j$ sends a ready message with all of that information to all parties.
The exact same procedure takes place when sending a ready message after receiving $f+1$ ready messages (except at this point it is not necessary to check that the commitment is correct).
Finally, after receiving $n-f$ ready messages, every nonfaulty party interpolates the corresponding points to a polynomial $p$, computes its coefficients $m_0,\ldots,m_\ell$, and outputs the message $M'=(m_0,\ldots,m_\ell)$. 

\begin{algorithm}\caption{RB}\label{alg:RB}
Code for party i:
\begin{algorithmic}[1] 
    \State $echoes[com]\gets\emptyset,readies[com]\gets\emptyset$ for each possible commitment $com$
    \State $c\gets\lceil\frac{\ell+1}{f+1}\rceil$
    \If{$i=d$}
        \State define the $\ell$-degree polynomial $p$ as follows: $p(x)=\sum_{k=0}^{\ell} m_i\cdot x^i$
        \State $\forall j\in[n]\ P_j\gets (p((j-1)\cdot c+1),\ldots,p(j\cdot c))$
        \State $P\gets(P_1,\ldots,P_n)$
        \State $com\gets \Commit(P)$ 
        \ForAll{$j\in[n]$}
            \State $\pi_j\gets \Open\Prove(P,j)$
            \State send party $j$ the message $\langle value, com, P_j,\pi_j\rangle$
        \EndFor
    \EndIf
    \Upon{receiving the first message of the form $\langle value, com, P_i, \pi_i\rangle$ from $d$ s.t. $\left|P_i\right|=c$}
        \If{$\Open\Prove(com,P_i,i,\pi_i)=1$}
            \State send $\langle echo, com, P_i, \pi_i\rangle$ to every party
        \EndIf
    \EndUpon
    \Upon{receiving the first $\langle echo, com, P_j, \pi_j\rangle$ messages from $j$ s.t. $\left|P_j\right|=c$}
        \If{$\Open\Verify(com,P_j,j,\pi_j)=1$}
            \State let $P_j=(p_{j,1},\ldots,p_{j,c})$
            \State $echoes[com]\gets echoes[com]\cup\{((j-1)\cdot c+k,p_{j,k})\}_{k\in[c]}$
            \If{$i$ hasn't sent a ready message and $\left|echoes[com]\right|\geq (n-f)\cdot c$}
                \State interpolate $\ell+1$ pairs from the set $echoes[com]$ to a polynomial $p'$
                \State $\forall j\in[n]\ P'_j\gets (p'((j-1)\cdot c+1),\ldots,p'(j\cdot c))$
                \State $P'\gets(P'_1,\ldots,P'_n)$
                \If{$\Commit(P')=com$}
                    \State $\pi_i\gets \Open\Prove(P',i)$
                    \State send $\langle ready, com, P'_i, \pi_i\rangle$ to every party
            \EndIf
        \EndIf
        \EndIf
    \EndUpon
    \Upon{receiving the first $\langle ready, com, P_j, \pi_j\rangle$ messages from $j$ s.t. $\left|P_j\right|c$}
        \If{$\Open\Verify(com,P_j,j,\pi_j)=1$}
            \State let $P_j=(p_{j,1},\ldots,p_{j,c})$
            \State $readies[com]\gets readies[com]\cup\{((j-1)\cdot c+k,p_{j,k})\}_{k\in[c]}$
            \If{$i$ hasn't sent a ready message and $\left|readies[com]\right|\geq (f+1)\cdot c$}
                \State interpolate $\ell+1$ pairs from the set $readies[com]$ to a polynomial $p'$
                \State $\forall j\in[n]\ P'_j\gets (p'((j-1)\cdot c+1),\ldots,p'(j\cdot c))$
                \State $P'\gets(P'_1,\ldots,P'_n)$
                \State $\pi_i\gets \Open\Prove(P',i)$
                \State send $\langle ready, com, P'_i, \pi_i\rangle$ to every party
            \EndIf
            \If{$\left|readies[com]\right|\geq (n-f)\cdot c$}
                \State interpolate $\ell+1$ pairs from the set $readies[com]$ to a polynomial $p'$
                \State let $m'_j$ be the $j'th$ coefficient in $p'$ and let $m'=(m'_0,\ldots,m'_{\ell})$
                \State \textbf{output} $m'$ and \textbf{terminate} \label{line:broadcastTermination}
            \EndIf
        \EndIf
    \EndUpon
\end{algorithmic}
\end{algorithm}

\begin{lemma}
    When a nonfaulty party tries to interpolate $\ell+1$ pairs in either the set $echoes[com]$ or $readies[com]$, there are indeed $\ell+1$ pairs in those sets. 
    Furthermore, for any nonfaulty party, if $(x,y),(x',y')\in echoes[com]$ or $(x,y),(x',y')\in readies[com]$, then either $x\neq x'$ or $(x,y)=(x',y')$.
\end{lemma}
\begin{proof}
    The proof only deals with the set $echoes[com]$.
    The exact same arguments can be made for $readies[com]$.
    A nonfaulty party tries to interpolate $\ell+1$ pairs in the set $echoes[com]$ when it finds that $\left|echoes[com]\right|\geq (n-f)\cdot c\geq (f+1)\cdot c$, for $c=\lceil\frac{\ell+1}{f+1}\rceil$.
    Substituting $c$: $\left|echoes[com]\right|\geq (f+1)\cdot \lceil\frac{\ell+1}{f+1}\rceil\geq (f+1)\cdot \frac{\ell+1}{f+1}=\ell+1$.
    For the second part of the lemma, a nonfaulty party only adds elements of the form $((j-1)\cdot c+k,p_{j,k})$ to $echoes[com]$ such that $k\in[c]$ after receiving an echo message from party $j$.
    However, for any pair $j,j'\in\mathbb{N}$ such that $j\neq j'$ and $k,k'\in[c]$, it cannot be the case that $(j-1)\cdot c+k=(j'-1)\cdot c+k'$ because the distance between $(j-1)\cdot c$ and $(j'-1)\cdot c$ is at least $c$.
\end{proof}

\begin{lemma}\label{lem:uniqueReady}
    If two nonfaulty parties $i,j$ send the messages $\langle ready, com, P_i,\pi_i\rangle$ and $\langle ready, com', P_j,\pi_j\rangle$, then $com=com'$.
\end{lemma}
\begin{proof}
    Let $i',j'$ be the first nonfaulty parties that sent messages with the values $com,com'$ respectively.
    Since $i'$ is the first nonfaulty party to send such a message, it couldn't have received a $\langle ready, com, P^k, \pi^k\rangle$ message from any party other than the $f$ faulty parties before sending such a message.
    The only other way for $i'$ to send such a message is after finding that $\left|echoes[com]\right|\geq (n-f)\cdot c$.
    Since $i'$ adds $c$ elements to $echoes[com]$ after receiving an $\langle echo, com, P_k,\pi_k\rangle$ from party $k$, this means it received such echo messages from $n-f$ parties.
    Similarly, $j'$ received an $\langle echo, com', P_k,\pi_k\rangle$ message from $n-f$ parties.
    Since $2(n-f)=n+(n-2f)\geq n+f+1$, $i'$ and $j'$ received those ready messages from at least $f+1$ common parties, and at least one of those parties is nonfaulty.
    Note that if some nonfaulty party sends an echo message it sends the same one to all parties, and thus $com=com'$.
\end{proof} 
\begin{lemma}\label{lem:readyProof}
    Let $c=\lceil\frac{\ell+1}{f+1} \rceil$ be defined as it is in the protocol.
    If a nonfaulty party $i$ sends the message $\langle ready, com, P_i,\pi_i\rangle$, then $\left|P_i\right|=c$ and $\Open\Verify(com,P_i,i,\pi_i)=1$.
\end{lemma}
\begin{proof}
    Party $i$ only sends the message $\langle ready, com, P_i,\pi_i\rangle$ if it finds that $\left|echoes[com]\right|\geq (n-f)\cdot c$ or if it finds that $\left|readies[com]\right|\geq (f+1)\cdot c$. 
    This can only happen as a result of receiving messages of the form $\langle echo, com, P_j,\pi_j\rangle$ from $n-f$ parties, or messages of the form $\langle ready, com, P_j,\pi_j\rangle$ from $f+1$ parties which pass verification tests.
    This is because whenever $i$ updates either of its $echoes$ or $readies$ sets, it adds exactly $c$ elements to them.
    If $i$ sent the message after receiving $n-f$ echo messages, then $i$ first interpolates $\ell+1$ of the points $(k,y_k)\in echoes[com]$ to a polynomial $p'$, for every $j\in[n]$ computes $P'_j=(p((j-1)\cdot c+1),\ldots, p(j\cdot c))$, sets $P'=(P'_1,\ldots,P'_n)$, and then checks that $\Commit(P')=com$.
    It then computes $\pi_i =\Open\Prove(P',i)$  and sends the message $\langle ready, com, P'_i,\pi_i\rangle$.
    Note that in that case, $com$ is indeed a commitment to $P'$, so $\Open\Verify(com,P'_i,i,\pi_i)=1$.
    On the other hand, if $i$ sent the message after receiving $f+1$ ready messages, then at least one of those messages was received from a nonfaulty party.
    Observe the first nonfaulty party $j$ that sent a $\langle ready, com, P_j, \pi_j\rangle$ message.
    No nonfaulty party has sent a ready message with the value $com$ at the time $j$ sent the message, so it could have only received ready messages with the value $com$ from the $f$ faulty parties, and thus $\left|readies[com]\right|\leq f\cdot c$.
    This means that before sending the message, it received $n-f$ messages of the form $\langle echo,com,P_k,\pi_k\rangle$, interpolated $\ell+1$ of the values in its $echoes[com]$ set to a polynomial $p'$, for every $l\in[n]$ computed $P'_l=(p'((l-1)\cdot c+1),\ldots p'(l\cdot c))$
    and found that $\Commit((P'_1,\ldots,P'_n))=com$.
    Since interpolating $\ell+1$ points always yields a polynomial of degree $\ell$ or less, this means that $com$ is a commitment to $n$ sets of $c$ points on the polynomial $p'$, which is of degree $\ell$ or less.
    Now, before sending the ready message, $i$ receives $f+1$ messages of the form $\langle ready, com, P_j, \pi_j\rangle$ such that $\forall \Open\Verify(com, P_j, j, \pi_j)=1$, and thus each such $P_j$ is a set of $c$ points on the polynomial $p'$.
    More precisely,  $P_j=(p'((j-1)\cdot c+1),\ldots,p'(j\cdot c))$.
    Party $i$ then interpolates $\ell+1$ of the pairs $(k,p'(k))\in readies[com]$ to a polynomial, and since $p'$ is of degree $\ell$ or less, that polynomial must be $p'$.
    Finally, $i$ computes $P'_j=(p'((j-1)\cdot c+1),\ldots,p'(j\cdot c))$ for every $j\in[n]$, $P'=(P'_1,\ldots,P'_n)$ and $\pi_i=(\Open\Prove(P',i))$.
    After computing those values, $i$ sends $\langle ready,com,P'_i,\pi_i\rangle$ to all parties.
    Clearly, in this case $\left|P'_i\right|=c$.
    In addition, since $com$ is a commitment to $P'$, it is also the case that $\Open\Verify(com, P'_i,i,\pi_i)=1$.
\end{proof}

\begin{theorem}
    Protocol $RB$ is a reliable broadcast protocol resilient to $f<\frac{n}{3}$ Byzantine parties.
\end{theorem}
\begin{proof}
    We will prove each property separately.
    In the proof, let $c=\lceil\frac{\ell+1}{f+1}\rceil$, as defined in the protocol.
    
    \textbf{Validity.} If the dealer is nonfaulty, it computes $p(x)=\sum_{k=0}^\ell m_k\cdot x^k$, computes $P_j=(p((j-1)\cdot c+1),\ldots p(j\cdot c))$ for every $j\in[n]$ and sets $P=(P_1,\ldots P_n)$.
    Afterwards, the dealer computes $com=\Commit(P)$ and then for every party $j$ it computes $\pi_j=\Open\Prove(P,j)$, and sends $j$ the message $\langle value, com, P_j, \pi_j\rangle$.
    Every nonfaulty party $j$ that sends an echo message does so after receiving the previous message and sends the message $\langle echo, com, P_j, \pi_j\rangle$.
    The nonfaulty parties send only one echo message, so every nonfaulty party receives no more than $f$ messages of the form $\langle echo, com', P_k, \pi_k\rangle$ with $com'\neq com$.
    Assume by way of contradiction some nonfaulty party sends a ready message $\langle ready, com', P', \pi'\rangle$ with $com'\neq com$, and let $j$ be the first nonfaulty party that doe so.
    Since $j$ is the first nonfaulty party to send such a message, at the time it sent the message it could have only received $\langle ready, com', P_k, \pi_k\rangle$ message with $com'\neq com$ from the $f$ faulty parties.
    Note that $i$ can either add exactly $c$ elements to $echoes[com']$ or no elements at all after receiving each of those messages, and thus at that time $\left|echoes[com']\right|\leq f\cdot c<(n-f)\cdot c$.
    This means $j$ must have sent the message as a result of finding that $\left|echoes[com]\right|\geq (n-f)\cdot c$, which could only happen after receiving $\langle echo, com', P_j, \pi_j\rangle$ messages from $n-f$ parties.
    However, $n-f\geq f+1$, so at least one of those parties is nonfaulty. 
    As discussed above, every nonfaulty party that sends an echo message sends one with the value $com\neq com'$, reaching a contradiction.
    Now observe some nonfaulty party $i$ that completes the protocol.
    Before doing so, it found that for some $com'$ $\left|readies[com']\right|\geq (n-f)\cdot c$.
    Party $i$ adds exactly $c$ elements to $readies[com']$ after receiving $\langle ready,com',P_j,\pi_j\rangle$ from some party $j$ that passes some verification tests.
    As shown above, no more than $f$ such messages could have been sent for any $com'\neq com$, in which case $\left|readies[com']\right|\leq c\cdot f<(n-f)\cdot c$, so $com'=com$.
    Any pair $((j-1)\cdot c+k,p_{j,k})$ that $i$ added to $readies[com]$ was added after finding that $\Open\Verify(com,P_j,j,\pi_j)=1$ and parsing $P_j$ as $(p_{j,1},\ldots,p_{j,c})$.
    Seeing as $com$ is a commitment to $(P_i,\ldots,P_n)$, it must be the case that $p_{j,k}=p((j-1)\cdot c+k)$.
    Now, before completing the protocol $i$ interpolates $\ell+1$ points $(m,p(m))$ on the polynomial $p$ of degree $\ell$ or less, and thus it computes $p$, then computes its coefficients $m_0,\ldots,m_\ell$, and finally outputs $M=(m_0,\ldots,m_\ell)$.
    
    \textbf{Agreement.} Let $i$, $j$ be two nonfaulty parties that output the messages $M,M'$ respectively.
    Before outputting those messages, $i$ found that for some value $com$ $\left|readies[com]\right|\geq (n-f)\cdot c$.
    This means that $i$ received a message of the form $\langle ready, com, P_k, \pi_k\rangle$ from $n-f$ parties such that for each one $\Open\Verify(com,P_k,k,\pi_k)=1$.
    The same can be said about $j$ having received similar messages with some value $com'$.
    Since $2(n-f)=n+(n-2f)\geq n+f+1$, $i$ and $j$ received the aforementioned messages from at least $f+1$ common parties, at least one of which is nonfaulty.
    Note that every nonfaulty party sends only one ready message to all parties throughout the protocol (with the same content), so $com=com'$.
    
    Observe the first nonfaulty party $i^*$ that sent a ready message with the commitment $com$.
    At that time, $i^*$ could have received no more than $f$ ready messages with the commitment $com$, and as discussed in the proof of the Validity property, this means that $\left|readies[com]\right|\leq f\cdot c<(f+1)\cdot c$.
    This means that $i^*$ decided to send the message after finding that $\left|echoes[com]\right|\geq (n-f)\cdot c$, interpolated $\ell+1$ of the values $(k,y_k)\in echoes[com]$ to a polynomial $p'$,
    computed $P'_k=(p'((k-1)\cdot c+1),\ldots,p'(k\cdot c))$ for every $k\in[n]$.
    It then set $P'=(P'_1,\ldots,P'_k)$ and found that $\Commit(P')=com$.
    Since interpolating $\ell+1$ points always yields a polynomial of degree $\ell$ or less, this means that $com$ is a commitment to $n$ sets of $c$ points on a polynomial of degree $\ell$ or less.
    Now, before outputting $M$ and $M'$, $i$ and $j$ found that $\left|readies[com]\right|\geq (n-f)\cdot c$.
    Again, as discussed above, this could only happen after
    receiving $n-f$ messages of the form $\langle ready, com, P_k, \pi_k\rangle$ such that $\left|P_k\right|==c$ and $\Open\Verify(com,P_k,k,\pi_k)=1$.
    $P_k$ is a commitment to a vector of $c$ points on $p'$, and thus $P_k=(p'((j-1)\cdot c+1),\ldots,p'(j\cdot c))$.
    Therefore, after receiving those messages, both $i$ and $j$ add $((k-1)\cdot c+l,p'((k-1)\cdot c+l))$ to $readies[com]$ for every $l\in[c]$.
    Those are the only values added to the set readies, so for every $(k,y_k)\in readies[com]$, $y_k=p'(k)$.
    Choosing any $\ell+1$ points $(k,y_k)\in readies[com]$, both $i$ and $j$ then compute the same polynomial $p'(x)=\sum_{i=0}^\ell m'_i\cdot x^i$, and output the same message $(m'_0,\ldots,m'_\ell)$.
    
    \textbf{Termination.}
    If the dealer is nonfaulty, it computes $p(x)=\sum_{k=0}^\ell m_k\cdot x^k$ and computes $P_j=(p((j-1)\cdot c+1),\ldots,p(j\cdot c))$ for every party $j\in[n]$.
    The dealer then sets $P=(P_1,\ldots,P_n)$, computes $com=\Commit(P)$ and then for every party $j$ it computes $\pi_j=\Open\Prove(P,j)$ and sends $j$ the message $\langle value, com, P_j, \pi_j\rangle$.
    Every nonfaulty party then receives that message, finds that $\left|P^j\right|=c$ and $\Open\Verify(com,P_j,j,\pi_j)=1$ and sends an $\langle echo, com, P_j, \pi_j\rangle$ message to all parties.
    Every nonfaulty eventually receives an $\langle echo, com, P_j,\pi_j\rangle$ message from every nonfaulty party, finds that the same conditions hold, parses $P_j$ as $(p_{j,1},\ldots,p_{j,c})$ and adds $((j-1)\cdot c+k,p_{j,k})$ to $echoes[com]$ for every $k\in[c]$.
    After doing that, every nonfaulty party $j$ finds that $\left|echoes[com]\right|\geq (n-f)\cdot c$, and if it hasn't sent a ready message yet, it interpolates $\ell+1$ points in $echoes[com]$ to a polynomial $p'$ and sends a ready message.
    From Lemma~\ref{lem:uniqueReady}, all of the ready messages sent by nonfaulty parties have the same value $com$, and from Lemma~\ref{lem:readyProof}, if a nonfaulty party sends a message $\langle ready, com, P_j, \pi_j\rangle$ then $\Open\Verify(com,P_j,j,\pi_j)=1$ and $\left|P_j\right|=c$ for every one of those messages. 
    Therefore, after receiving each of those messages, every nonfaulty party updates its $readies[com]$ set and adds $c$ elements to it.
    After adding $c$ such elements for every nonfaulty $j$, every nonfaulty party finds that $\left|readies[com]\right|\geq (n-f)\cdot c$, performs some local computations, and completes the protocol.
    
    For the second part of the property, if some nonfaulty party completes the protocol it received $n-f$ messages of the form $\langle ready,com,P_j,\pi_j\rangle$ with the same value $com$ such that $\Open\Verify(com,P_j,j,\pi_j)=1$ and $\left|P^j\right|=c$.
    Out of those $n-f$ messages, at least $n-2f\geq f+1$ were sent by nonfaulty parties.
    Every nonfaulty party eventually receives those $f+1$ messages, finds the same conditions hold, and adds $c$ elements to $readies[com]$.
    After adding $c$ elements for every one of those $f+1$ parties, every nonfaulty $i$ sees that $\left|readies[com]\right|\geq (f+1)\cdot c$, performs some local computations and sends a message $\langle ready, com, P_i,\pi-i\rangle$ itself, if it hasn't done so earlier.
    From Lemma~\ref{lem:uniqueReady}, every nonfaulty party that sent a ready message previously also sent one with the same value $com$.
    From Lemma~\ref{lem:readyProof}, $\Open\Verify(com,P_i,i,\pi_i)=1$ and $\left|P^i\right|=c$, so after receiving those messages, every nonfaulty party adds $c$ elements to $readies[com]$.
    Finally, after adding $c$ elements to $readies[com]$ for every nonfaulty party, every nonfaulty party finds that $\left|readies[com]\right|\geq (n-f)\cdot c$, performs some local computations, and completes the protocol.
\end{proof}

\subsection{Proof of Theorem~\ref{thm:RBComplexity}}\label{sec:broadcast:proofRBComplexity}
\begin{proof}
    Let the number of words in the message be $\ell+1$.
    Throughout the protocol, the dealer starts by sending a single message to every party containing a commitment and $O(\frac{\ell}{n})$ words and proofs.
    Then, every party sends at most one echo message and one ready message containing a commitment, a proof and a set containing $O(\frac{\ell}{n})$ words.
    Overall, there are $O(n^2)$ messages, each containing $c$ words for the commitment, $p$ words for the proof and $O(\frac{\ell}{n})$ additional words. 
    This yields a total of $O(n^2\cdot (c+p)+\frac{\ell}{n}\cdot n^2)=O(n^2\cdot (c+p)+\ell\cdot n)$ words.
\end{proof}

The protocol can trivially be turned into a Validated Reliable Broadcast protocol, $VRB$, by only having parties output $m'$ in line~\ref{line:broadcastTermination} after checking that $\validate(m')=1$.
This clearly makes the additional part of the Validity property hold, and doesn't change the rest of the proof for the Validity and Correctness properties.
In the proof of the Termination property, first we can note that if some nonfaulty party were to output a message $M'$ when the dealer is nonfaulty, then from the Validity property it must be the case that $M'=M$.
This means that if the dealer does have an input $M$ such that $\validate(M)=1$, all nonfaulty parties would reach that point in the protocol, see that $\validate(M)=1$, and terminate.
In addition, if some nonfaulty party completes the protocol, it must have output some value some $M'$ such that $\validate(M')=1$.
Using the exact same arguments as the one in the proof of the Termination property, all nonfaulty parties eventually reach the end of the protocol.
From the Correctness property, they reach the end of the protocol with the same message $M'$, and thus when checking if $\validate(M')=1$ they all see that the condition holds and output $M'$.

\end{document}